\documentclass[journal,comsoc]{IEEEtran}
\usepackage[utf8]{inputenc} 
\usepackage[T1]{fontenc}
\usepackage{url}
\usepackage{ifthen}
\usepackage[cmex10]{amsmath}
\usepackage{array} 

\usepackage{enumerate}
\usepackage{amssymb}
\usepackage{amsmath} 

\usepackage{algorithm}
\usepackage{algorithmicx}
\usepackage{algpseudocode}

\usepackage{mathrsfs}
\usepackage{bm}
\usepackage{tikz}
\usetikzlibrary{arrows}
\usepackage{subfigure}
\usepackage{graphicx,booktabs,multirow}
\usepackage{epstopdf}
\usepackage{subfigure}
\usepackage{multirow}
\usepackage{tabularx} 
\usepackage{booktabs}

\usepackage{pifont}

\definecolor{colorhkust}{RGB}{20,43,140}
\definecolor{colortsinghua}{RGB}{116,52,129}
\definecolor{color1}{RGB}{128,0,0}

\usepackage{amsthm}
\usepackage{cite}
\newtheorem{thm}{Theorem}[section]

\newtheorem{prop}{Proposition}[section]

\theoremstyle{definition}
\newtheorem{defn}{Definition}[section]

\theoremstyle{remark}
\newtheorem*{rem}{Remark}

\begin{document}

        \title{Graph Neural Networks for Scalable Radio Resource Management: Architecture Design and Theoretical Analysis}

      \author{Yifei Shen, \textit{Student Member}, \textit{IEEE}, Yuanming Shi, \textit{Member}, \textit{IEEE}, Jun Zhang, \textit{Senior Member}, \textit{IEEE}, and Khaled B. Letaief, \textit{Fellow}, \textit{IEEE}
      	\thanks{The materials in this paper were presented in part at the IEEE Global Communications Conference (Globecom) Workshops, 2019 \cite{shen2019graph}. This work was supported by the General Research Fund (Project No. 16210719) from the Research Grants Council of Hong Kong.
      		
      	Y. Shen and K. B. Letaief are with the Department of Electronic and Computer Engineering, Hong Kong University of Science and Technology, Hong Kong (E-mail: \{yshenaw, eekhaled\}@ust.hk). K. B. Letaief is also with Peng Cheng Lab in Shenzhen. Y. Shi is with the School of Information Science and Technology, ShanghaiTech University, Shanghai 201210, China (E-mail:  shiym@shanghaitech.edu.cn). J. Zhang is with the Department of Electronic and Information Engineering, The Hong Kong Polytechnic University, Hong Kong (E-mail: jun-eie.zhang@polyu.edu.hk). (The corresponding author is J. Zhang.)}

}
        
        \maketitle

\begin{abstract}
Deep learning has recently emerged as a disruptive technology to solve challenging radio resource management problems in wireless networks. However, the neural network architectures adopted by existing works suffer from poor scalability and generalization, and lack of interpretability. A long-standing approach to improve scalability and generalization is to incorporate the structures of the target task into the neural network architecture. In this paper, we propose to apply graph neural networks (GNNs) to solve large-scale radio resource management problems, supported by effective neural network architecture design and theoretical analysis. Specifically, we first demonstrate that radio resource management problems can be formulated as graph optimization problems that enjoy a universal permutation equivariance property. We then identify a family of neural networks, named \emph{message passing graph neural networks} (MPGNNs). It is demonstrated that they not only satisfy the permutation equivariance property, but also can generalize to large-scale problems, while enjoying a high computational efficiency. For interpretablity and theoretical guarantees, we prove the equivalence between MPGNNs and a family of distributed optimization algorithms, which is then used to analyze the performance and generalization of MPGNN-based methods. Extensive simulations, with power control and beamforming as two examples, demonstrate that the proposed method, trained in an unsupervised manner with unlabeled samples, matches or even outperforms classic optimization-based algorithms without domain-specific knowledge. Remarkably, the proposed method is highly scalable and can solve the beamforming problem in an interference channel with $1000$ transceiver pairs within $6$ milliseconds on a single GPU.

\end{abstract}

\section{Introduction}
Radio resource management, e.g., power control \cite{chiang2008power} and beamforming \cite{bjornson2014optimal}, plays a crucial role in wireless networks. Unfortunately, many of these problems are non-convex and computationally challenging. Moreover, they need to be solved in a real-time manner given the time-varying wireless channels and the latency requirement of many mobile applications. Great efforts have been put forward to develop effective algorithms for these challenging problems. Existing algorithms are mainly based on convex optimization approaches \cite{shi2015large,shi15largescale}, which have a limited capability in dealing with non-convex problems and scale poorly with the problem size. Problem specific algorithms can be developed, which, however, is a laborious process and requires much problem specific knowledge.

Inspired by the recent successes of deep learning in many application domains, e.g., computer vision and natural language processing \cite{goodfellow2016deep}, researchers have attempted to apply deep learning based methods, particularly, ``learning to optimize'' approaches, to solve difficult optimization problems in wireless networks \cite{sun2018learning,lee2018deep,liang2018towards,cui2018spatial,shen2018lora,xia2019deep,lee2019graph,eisen2019optimal,dong2020faster}. The goal of such methods is to achieve near-optimal performance in a real-time manner without domain knowledge, i.e., to automate the algorithm design process. There are two common paradigms on this topic \cite{bengio2018machine,liang2019deep}. The first one is ``end-to-end learning'', which directly employs a neural network to approximate the optimal solution of an optimization problem. For example, in \cite{sun2018learning}, to solve the power control problem, a multi-layer perceptron (MLP) was used to approximate the input-output mapping of the classic weighted minimum mean square error (WMMSE) algorithm \cite{Shi2011An} to speed up the computation. The second paradigm is ``learning alongside optimization'', which replaces an ineffective policy in a traditional algorithm with a neural network. For example, an MLP was utilized in \cite{shen2018lora} to replace the pruning policy in the branch-and-bound algorithm.  Accordingly, significant speedup and performance gain in the access point selection problem was achieved compared with the optimization-based methods in \cite{shi2014group,shi2018enhanced}.   

A key design ingredient underlying both paradigms of ``learning to optimize'' is the neural network architecture. Most of the existing works adopt MLPs \cite{sun2018learning,liang2018towards,shen2018lora,sun2020data} or convolutional neural networks (CNNs) \cite{lee2018deep,xia2019deep}. These architectures are inherited from the ones developed for image processing tasks and thus are not tailored to problems in wireless networks. Although near-optimal performance is achieved for small-scale wireless networks, they fail to exploit the wireless network structure and thus suffer from poor scalability and generalization in large-scale radio resource management problems. Specifically, the performance of these methods degrades dramatically when the wireless network size becomes large. For example, it was shown in \cite{sun2018learning} that the performance gap to the WMMSE algorithm is $2\%$ when $K = 10$ and it becomes $12\%$ when $K = 30$. Moreover, these methods generalize poorly when the number of agents in the test dataset is larger than that in the training dataset. In dense wireless networks, resource management may involve thousands of users simultaneously and the number of users changes dynamically, thus, making the wide application of these learning-based methods very difficult.

A long-standing idea to improve scalability and generalization is to incorporate the structures of the target task into the neural network architecture \cite{bengio2018machine,ravanbakhsh2017equivariance,xu2019what,sun2020data}. A prominent example is the development of CNNs for computer vision, which is inspired by the fact that the neighbor pixels of an image are useful when they are considered together \cite{brendel2019approximating}. This idea has also been successfully applied in many applications, e.g., visual reasoning \cite{xu2019what}, combinatorial optimization \cite{li2018combinatorial}, and route planning \cite{zhuang2020towards}. To achieve better scalability of learning-based radio resource management, structures in a single-antenna system with homogeneous agents have recently been exploited for effective neural network architecture design \cite{cui2018spatial,eisen2019optimal}. In static channels, observing that channel states are deterministic functions of users' geo-locations in a 2D Euclidean space, spatial convolution was developed in \cite{cui2018spatial}, which is applicable in wireless networks with thousands of users but cannot handle fading channels. With fading channels, it was observed that the channel matrix can be viewed as the adjacency matrix of a graph \cite{eisen2019optimal}. From this perspective, a random edge graph neural network (REGNN) operating on such a graph was developed, which inhibits a good generalization property when the number of users in the wireless networks changes. However, in a multi-antenna system or a single-antenna system with heterogeneous agents, the channel matrix no longer fits the form of an adjacency matrix and the REGNN cannot be applied. 

In this paper, we address the limitations of existing works by modeling wireless networks as \emph{wireless channel graphs} and develop neural networks to exploit the graph topology. Specifically, we treat the agents as nodes in a graph, communication channels as directed edges, agent specific parameters as node features, and channel related parameters as edge features. Subsequently, low-complexity neural network architectures operating on wireless channel graphs will be proposed.

Existing works (e.g., \cite{sun2018learning,shen2018lora,lee2019graph}) also have another major limitation, namely, they treat the adopted neural network as a black box. Despite the superior performance in specific applications, it is hard to interpret what is learned by the neural networks. To ensure reliability, it is crucial to understand when the algorithm works and when it fails. Thus, a good theoretical understanding is demanded for the learning-based radio resource management methods. Compared with learning-based methods, conventional optimization-based methods are well-studied. This inspires us to build a relationship between these two types of methods. In particular, we shall prove the \emph{equivalence} between the proposed neural networks and a favorable family of optimization-based methods. This equivalence will allow the development of tractable analysis for the performance and generalization of the learning-based methods through the study of their equivalent optimization-based methods.

\subsection{Contributions}
In this paper, we develop scalable learning-based methods to solve radio resource management problems in dense wireless networks. The major contributions are summarized as follows:

\begin{enumerate}
	\item We model wireless networks as wireless channel graphs and formulate radio resource management problems as graph optimization problems. We then show that a permutation equivariance property holds in general radio resource management problems, which can be exploited for effective neural network architecture design.
	
	\item We identify a favorable family of neural networks operating on wireless channel graphs, namely MPGNNs. It is shown that MPGNNs satisfy the permutation equivariance property, and have the ability to generalize to large-scale problems while enjoying a high computational efficiency. 
	
	\item For an effective implementation, we propose a \emph{wireless channel graph convolution network} (WCGCN) within the MPGNN class. Besides inheriting the  advantages of MPGNNs, the WCGCN enjoys several unique advantages for solving radio resource management problems. First, it can effectively exploit both agent-related features and channel-related features effectively. Second, it is insensitive to the corruptions of features, e.g., channel state information (CSI), implying that they can be applied with partial and imperfect CSI.
	
	\item To provide interpretability and theoretical guarantees, we prove the equivalence between MPGNNs and a family of distributed optimization algorithms, which include many classic algorithms for radio resource management, e.g., WMMSE \cite{Shi2011An}. Based on this equivalence, we analyze the performance and generalization of MPGNN-based methods in the weighted sum rate maximization problem.

	\item We test the effectiveness of WCGCN for power control and beamforming problems, training with unlabeled data. Extensive simulations will demonstrate that the proposed WCGCN matches or outperforms classic optimization-based algorithms without domain knowledge, and with significant speedups. Remarkably, WCGCN can solve the beamforming problem with $1000$ users within $6$ milliseconds on a single GPU.\footnote{The codes to reproduce the simulation results are available on https://github.com/yshenaw/GNN-Resource-Management.}
	
\end{enumerate}

\subsection{Notations} Throughout this paper, superscripts $(\cdot)^H$, $(\cdot)^T$, $(\cdot)^{-1}$ denote conjugate transpose, transpose, inverse, respectively. The symbol $X_{(i_1, \cdots, i_n)}$ denotes an element in tensor $X$ indexed by $i_1, \cdots, i_n$. For example, $\bm{X}_{(2,3)}$ is the element in the second row third column in matrix $\bm{X}$. The set symbol $\{ \}$ in this paper denotes a multiset. A multiset is a $2$-tuple $X=(S,m)$ where $S$ is the underlying set of $X$ that is formed from its distinct elements, and $m: S \rightarrow \mathbb{N}_{\geq 1}$ gives the multiplicity of elements. For example, $\{a,a,b\}$ is a multiset where element $a$ has multiplicity $2$ and element $b$ has multiplicity $1$.

\section{Graph Modeling of Wireless Networks} \label{sec:graph}
In this section, we model wireless networks as graphs, and formulate radio resource management problems as graph optimization problems. Key properties of radio resource management problems will be identified, which will then be exploited to design effective neural network architectures.

\subsection{Directed Graphs and Permutation Equivariance Property}
A directed graph can be represented as an order pair $G = (V,E)$, where $V$ is the set of nodes and $E$ is the set of edges. The \emph{adjacency matrix} of a graph is an $n \times n$ matrix $\bm{A} \in \{0,1\}^{n \times n}$, where $\bm{A}_{i,j} = 1$ if and only if $(i,j) \in E$ for all $i,j \in V$. Let $[n] = \{1, \cdots, n \}$ and we denote the permutation operator as $\pi:[n] \rightarrow [n]$. 
Given the permutation $\pi$ and a graph adjacency matrix $\bm{A}$, the permutation of nodes is denoted by $\pi \star \bm{A}$ and defined as
\begin{align*}
(\pi \star \bm{A})_{(\pi(i_1),\pi(i_2)) } = \bm{A}_{ (i_1,i_2) },
\end{align*} 
for index $1 \leq i_1,i_2 \leq |V|$. Two graphs $\bm{A}$ and $\bm{B}$ are said to be \emph{isomorphic} if there is a permutation $\pi$ such that $\pi \star \bm{A} = \bm{B}$, and this relationship is denoted by $\bm{A} \cong \bm{B}$.

We now introduce optimization problems defined on directed graphs, and identify their permutation invariance and equivariance properties. We assign each node $v_i \in V$ an optimization variable $\gamma_i \in \mathbb{R}$. We denote the optimization variable as $\bm{\gamma} = [\gamma_1,\cdots,\gamma_{|V|}]^T$ and the permutation of the optimization variable as 
\begin{align*}
(\pi \star \bm{\gamma})_{(\pi(i_1)) } = \bm{\gamma}_{(i_1)}. 
\end{align*}

The optimization problem defined on a graph $\bm{A}$ can be written as 
\begin{equation} \label{eq:sg_opt}
\begin{aligned}
&\mathscr{Q}:\underset{\bm{\gamma}}{\text{minimize}}
& & g(\bm{\gamma},\bm{A})
& \text{ subject to }
& & Q(\bm{\gamma},\bm{A}) \leq 0,
\end{aligned}
\end{equation}
where $g(\cdot,\cdot)$ represents the objective function and $Q(\cdot,\cdot)$ represents the constraint.

As $\bm{A} \cong \pi \star \bm{A}$, optimization problems defined on graphs have the permutation invariance property as stated below.

\begin{prop} (Permutation invariance) \label{prop:pi_dg}
	The optimization problem defined in (\ref{eq:sg_opt}) has the following property
	\begin{align*}
		g(\bm{\Gamma},\bm{A}) = g(\pi \star \bm{\gamma},\pi \star \bm{A}), \quad Q(\bm{\gamma},\bm{A}) = Q(\pi \star \bm{\gamma},\pi \star \bm{A}),
	\end{align*}
	for any permutation operator $\pi$.
\end{prop}
\begin{proof}
	Since adjacency matrices $\bm{A}$ and $\pi \star \bm{A}$ represent the same graph, permuting $\bm{\gamma}$ and $\bm{A}$ simultaneously is simply a reordering of the variables. As a result, we have $g(\bm{\Gamma},\bm{A}) = g(\pi \star \bm{\gamma},\pi \star \bm{A})$ and $Q(\bm{\gamma},\bm{A}) = Q(\pi \star \bm{\gamma},\pi \star\bm{A})$.
\end{proof}

The permutation invariance property of the objective value and constraint leads to the corresponding property of sublevel sets. We first define the sublevel sets.

\begin{defn} (Sublevel sets)
	The $\alpha$ sublevel set of a function $f:\mathbb{C}^n \rightarrow \mathbb{R}$ is defined as 
	\begin{align*}
		\mathcal{R}^{\alpha}_f = \{x \in \text{dom} f | f(x) \leq \alpha  \},
	\end{align*}
	where $\text{dom} f$ is the feasible domain.
\end{defn}

Denote the optimal objective value of (\ref{eq:sg_opt}) as $z^*$, and the set of $\epsilon$-accurate solutions as $\mathcal{R}^{z^*+\epsilon}$. Thus, the properties of sublevel sets imply the properties of near-optimal solutions. Specifically, the permutation invariance property of the objective function implies the permutation equivariance property of the sub-level sets, which is stated in the next proposition.

\begin{prop} (Permutation equivariance) \label{prop:pe_dg}
	Denote $\mathcal{R}^{\alpha}_g$ as the sublevel set of $g(\cdot,\cdot)$ in (\ref{eq:sg_opt}), and define $F:\bm{A} \mapsto \mathcal{R}_g^{\alpha}$. Then,
	\begin{align*}
		F(\pi \star \bm{A}) = \{\pi \star \bm{\gamma}|\bm{\gamma} \in \mathcal{R}^{\alpha}_g\},
	\end{align*}
	for any permutation operator $\pi$.
\end{prop}

\begin{rem}
   The permutation equivariance property of sublevel sets is a direct result of the permutation invariance in the objective function. Please refer to Appendix \ref{proof:pe_dg} for a detailed proof.
\end{rem}

In the next subsection, by modeling wireless networks as graphs, we show that the permutation equivariance property is universal in radio resource management problems.

\subsection{Wireless Network as a Graph}
\begin{figure*}[htb]
	\centering
	\includegraphics[width=0.9\textwidth]{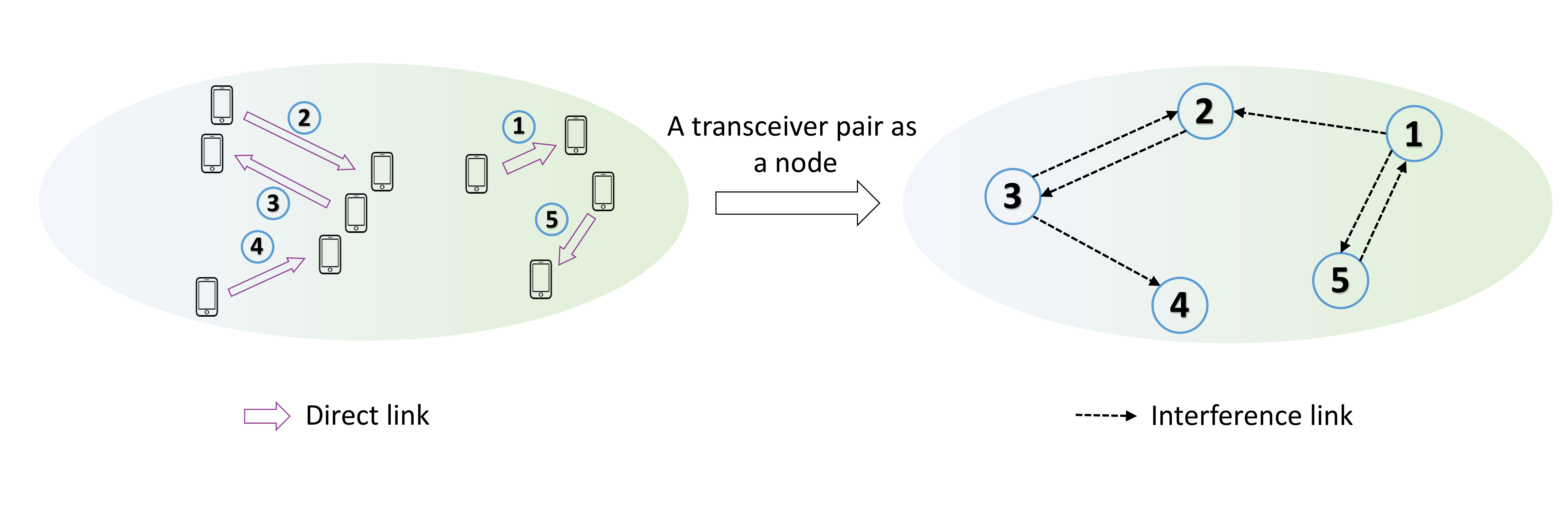}
	\caption{An illustration of graph modeling of a $K$-user interference channel.}
	\label{fig:d2d}
\end{figure*}
A wireless network can be modeled as a directed graph with node and edge features. Naturally, we treat each agent of a wireless network, e.g., a mobile user or a base station, as a node in the graph. An edge is drawn from node $i$ to node $j$ if there is a direct communication or interference link with node $i$ as the transmitter and node $j$ as the receiver. The node feature incorporates the properties of the agent, e.g., users' weights in the weighted sum rate maximization problem \cite{Shi2011An}. The edge feature includes the properties of the corresponding channel, e.g., a scalar (or matrix) to denote the channel state of a single-antenna (or multi-antenna) system. We call these graphs generated by the wireless network topology as \emph{wireless channel graphs}. Formally, a wireless channel graph is an ordered tuple $G = (V,E,s,t)$, where $V$ is the set of nodes, $E$ is the set of edges, $s: V \rightarrow \mathbb{C}^{d_1}$ maps a node to its feature, and $t: E \rightarrow \mathbb{C}^{d_2}$ maps an edge to its feature. Denote $V=\{v_1,v_2,\cdots,v_{|V|}\}$. Also define the node feature matrix as $\bm{Z} \in \mathbb{C}^{|V| \times d_1}$ with $\bm{Z}_{(i,:) } = s(v_i)$, and the adjacency feature tensor $A \in \mathbb{C}^{|V| \times |V| \times d_2}$ as
\begin{align} \label{eq:edge_fea}
A_{ (i,j,:) } = \left\{
\begin{aligned}
&\bm{0}, &&\text{if } \{i,j\} \notin E \\
&t(\{i,j\}) &&\text{otherwise},
\end{aligned}
\right. 
\end{align}
where $\bm{0}$ is a zero vector in $\mathbb{C}^{d_2}$. Given the permutation $\pi$, a graph $G$ with its node feature matrix $\bm{Z}$ and adjacency feature tensor $A$, the permutation of nodes is denoted by $(\pi \star \bm{Z},\pi \star A)$ and defined as
\begin{align*}
(\pi \star \bm{Z})_{ (\pi(i_1),:) } = \bm{Z}_{ (i_1,:) }, \quad (\pi \star A)_{ (\pi(i_1),\pi(i_2),:) } = A_{ (i_1,i_2,:) }.
\end{align*}  

We assign each node $v_i \in V$ an optimization variable $\bm{\gamma}_i \in \mathbb{C}^{n}$. Let $\bm{\Gamma} = [\bm{\gamma}_1,\cdots,\bm{\gamma}_{|V|}]^T \in \mathbb{C}^{|V| \times n}$, then an optimization problem defined on a wireless channel graph can be written as 

\begin{equation} \label{eq:cg_opt}
\begin{aligned}
&\mathscr{P}:\underset{\bm{\Gamma}}{\text{minimize}}
& & g(\bm{\Gamma},\bm{Z},A)
& \text{ subject to }
& & Q(\bm{\Gamma},\bm{Z},A) \leq 0
\end{aligned},
\end{equation}
where $g(\cdot,\cdot,\cdot)$ denotes the objective function and $Q(\cdot, \cdot, \cdot)$ denotes the constraint.

Next we elaborate the properties of the radio resource management problems on the wireless channel graphs. Without node features or edge features, a wireless channel graph is a directed graph. As a result, the properties of wireless channel graphs follow the properties of directed graphs. We elaborate the permutation equivariance property of problems on wireless channel graphs next. 
 Define the permutation of optimization variable as 
	\begin{align*}
	(\pi \star \bm{\Gamma})_{ (\pi(i_1),:) }  = \bm{\Gamma}_{ (i_1,:) }.
	\end{align*}

Similar to optimization problems on directed graphs, the ones defined on wireless channel graphs have the permutation invariance property. As a result, the sub-level sets of $g(\cdot,\cdot,\cdot)$ in (\ref{eq:cg_opt}) also have the permutation equivariance property, which is stated below.

\begin{prop} (Permutation equivariance) \label{prop:pe_cg}
	Let $\mathcal{R}^{\alpha}_g$ denote the sublevel set of $g(\cdot,\cdot,\cdot)$ in (\ref{eq:cg_opt}), and define $F:(\bm{Z},A) \mapsto \mathcal{R}_g^{\alpha}$. Then, 
	\begin{align*}
		F((\pi \star \bm{Z},\pi \star A)) = \{\pi \star \bm{\Gamma}|\bm{\Gamma} \in \mathcal{R}^{\alpha}_g\}.
	\end{align*}
	for any permutation operator $\pi$.
\end{prop}

\begin{rem}
	This result establishes a general permutation equivariance property for radio resource management problems. Proposition \ref{prop:pe_cg} is reduced to the results in \cite{eisen2019optimal} if $\bm{Z}$ is an all one matrix and $A \in \mathbb{R}^{|V| \times |V| \times 1}$. By modeling the node heterogeneity into $\bm{Z}$, Proposition \ref{prop:pe_cg} is applicable to heterogeneous agents. By introducing adjacency feature tensor instead of using adjacency matrix, this graph modeling technique can incorporate multi-antenna channel states. The proof is the same as Proposition \ref{prop:pe_dg} by simply changing notations.
\end{rem}

\subsection{Graph Modeling of $K$-user Interference Channels}

In this subsection, as a specific example, we present graph modeling of a classic radio resource management problem, i.e., beamforming for weighted sum rate maximization in a $K$-user interference channel. It will be used as the main test setting for the theoretical study in Section \ref{sec:per} and simulations in Section \ref{sec:exp}. There are in total $K$ transceiver pairs where each transmitter is equipped with $N_t$ antennas and each receiver is equipped with a single antenna. Let $\bm{v}_{k}$ denote the beamformer of the $k$-th transmitter. The received signal at receiver $k$ is $\bm{y}_{k} = \bm{h}_{k,k}^H \bm{v}_{k} s_{k} + \sum_{j\neq k}^K  \bm{h}_{j,k}^H\bm{v}_{j} s_{j} + n_{k}$, where $\bm{h}_{j,k} \in \mathbb{C}^{N_t}$ denotes the channel state from transmitter $j$ to receiver $k$ and $n_{k} \in \mathbb{C}$ denotes the additive noise following the complex Gaussian distribution $\mathcal{CN}(0,\sigma^2_{k})$.

The signal-to-interference-plus-noise ratio (SINR) for receiver $k$ is given by 
\begin{align} \label{eq:sinr}
\text{SINR}_k = \frac{|\bm{h}_{k,k}^H \bm{v}_{k}|^2}{ \sum_{j\neq k}^K |\bm{h}_{j,k}^H\bm{v}_{j}|^2 + \sigma_{k}^2}.
\end{align}

Denote $\bm{V} = [\bm{v}_{1}, \cdots, \bm{v}_{k}]^T \in \mathbb{C}^{K \times N_t}$ as the beamforming matrix. The objective is to find the optimal beamformer to maximize the weighted sum rate, and the problem is formulated as

\begin{equation}\label{eq:sys_mod}
\begin{aligned}
&\underset{\bm{V}}{\text{maximize}}
& & \sum_{k=1}^{K} w_{k} \log_2 \left(1+ \text{SINR}_{k} \right) \\
& \text{subject to}
& & \|\bm{v}_{k}\|_2^2 \leq P_{\text{max} }, \forall k,
\end{aligned}
\end{equation}
where $w_k$ is the weight for the $k$-th pair.

\paragraph{Graph Modeling} We view the $k$-th transceiver pair as the $k$-th node in the graph. As distant agents cause little interference, we draw a directed edge from node $j$ to node $k$ only if the distance between transmitter $j$ and receiver $k$ is below a certain threshold $D$. An illustration of such a graph modeling is shown in Fig. \ref{fig:d2d}. The node feature matrix $\bm{Z} \in \mathbb{C}^{|V| \times (N_t + 2)}$ is given by
\begin{align*}
	\bm{Z}_{ (k,:) } = [ \bm{h}_{k,k}, w_k, \sigma_k^2]^T,
\end{align*}
and the adjacency feature array $A \in \mathbb{C}^{|V| \times |V| \times N_t}$ is given by
\begin{align*}
A_{ (j,k,:) } = \left\{
\begin{aligned}
&\bm{0}, &&\text{if } \{j,k\} \notin E \\
&\bm{h}_{j,k} &&\text{otherwise},
\end{aligned}
\right. 
\end{align*}
where $\bm{0} \in \mathbb{C}^{N_t}$ is a zero vector. With notations $\bm{Z}$, $A$, and $\bm{V}$, SINR can be written as 
\begin{align*}
\text{SINR}_k = \frac{|\bm{Z}_{ (k,1:N_t) }^H \bm{v}_{k}|^2}{ \sum_{j\neq k}^K |\bm{A}_{ (j,k,:) }^H\bm{v}_{j}|^2 + \bm{Z}_{(k,N_t+2)} }.
\end{align*}
and (\ref{eq:sys_mod}) can be written as 
\begin{equation}\label{eq:sys_mod_graph}
\begin{aligned}
&\underset{\bm{V}}{\text{maximize}}
& & g(\bm{V},\bm{Z},A) = \sum_{k=1}^{K} \bm{Z}_{ (k, N_t+1) } \log_2 \left(1+ \text{SINR}_{k} \right) \\
& \text{subject to}
& & Q(\bm{V},\bm{Z},A) = \|\bm{v}_{k}\|_2^2 - P_{\text{max} }\leq 0, \forall k,
\end{aligned}
\end{equation}

Problem (\ref{eq:sys_mod_graph}) has the permutation equivariance property with respect to $\bm{V}$, $\bm{Z}$, and $A$ as shown in Proposition \ref{prop:pe_cg}. To solve this problem efficiently and effectively, the adopted neural network should exploit the permutation equivariance property, and incorporate both node features and edge features. We shall develop an effective neural network architecture to achieve this goal in the next section.

\section{Neural Network Architecture Design for Radio Resource Management}
In this section, we endeavor to develop a scalable neural network architecture for radio resource management problems. A favorable family of GNNs, named, \emph{message passing graph neural networks}, will be identified. The key properties and effective implementation will also be discussed.

\subsection{Optimizing Wireless Networks via Graph Neural Networks}
Most of existing works on ``learning to optimize'' approaches to solve problems in wireless networks adopted MLPs as the neural network architecture \cite{sun2018learning,liang2018towards,shen2018lora}. Although MLPs can approximate well-behaved functions \cite{hornik1989multilayer}, they suffer from poor performance in data efficiency, robustness, and generalization. A long-standing idea for improving the performance and generalization is to incorporate the structures of the target task into the neural network architecture. In this way, there is no need for the neural network to learn such structures from data, which leads to a  more efficient training, and better generalization empirically \cite{ravanbakhsh2017equivariance,sun2020data,pratik2020re,eisen2019optimal} and provably \cite{xu2019what}.

As discussed above, the structures of radio resource management problems can be formulated as optimization problems on wireless channel graphs, which enjoy the permutation equivariance property. In machine learning, there are two classes of neural networks that are able to exploit the permutation equivariance property, i.e., graph neural networks (GNNs) \cite{wu2019comprehensive} and Deep Sets \cite{zaheer2017deep}. Compared with Deep Sets, GNNs not only respect the permutation equivariance property but can also model the interactions among the agents. In wireless networks, the agents interact with each other through channels. Thus, GNNs are more favorable than Deep Sets in wireless networks. This motivates us to adopt GNNs to solve radio resource management problems.

\subsection{Message Passing Graph Neural Networks} \label{sec:mpgnn}
\begin{figure*}[h] 
	\centering
	 
	\subfigure[An illustration of CNNs. In each layer, each pixel convolves itself and neighbor pixels.]{
		\begin{minipage}[t]{0.47\linewidth}
			\centering
			\includegraphics[width=1\linewidth]{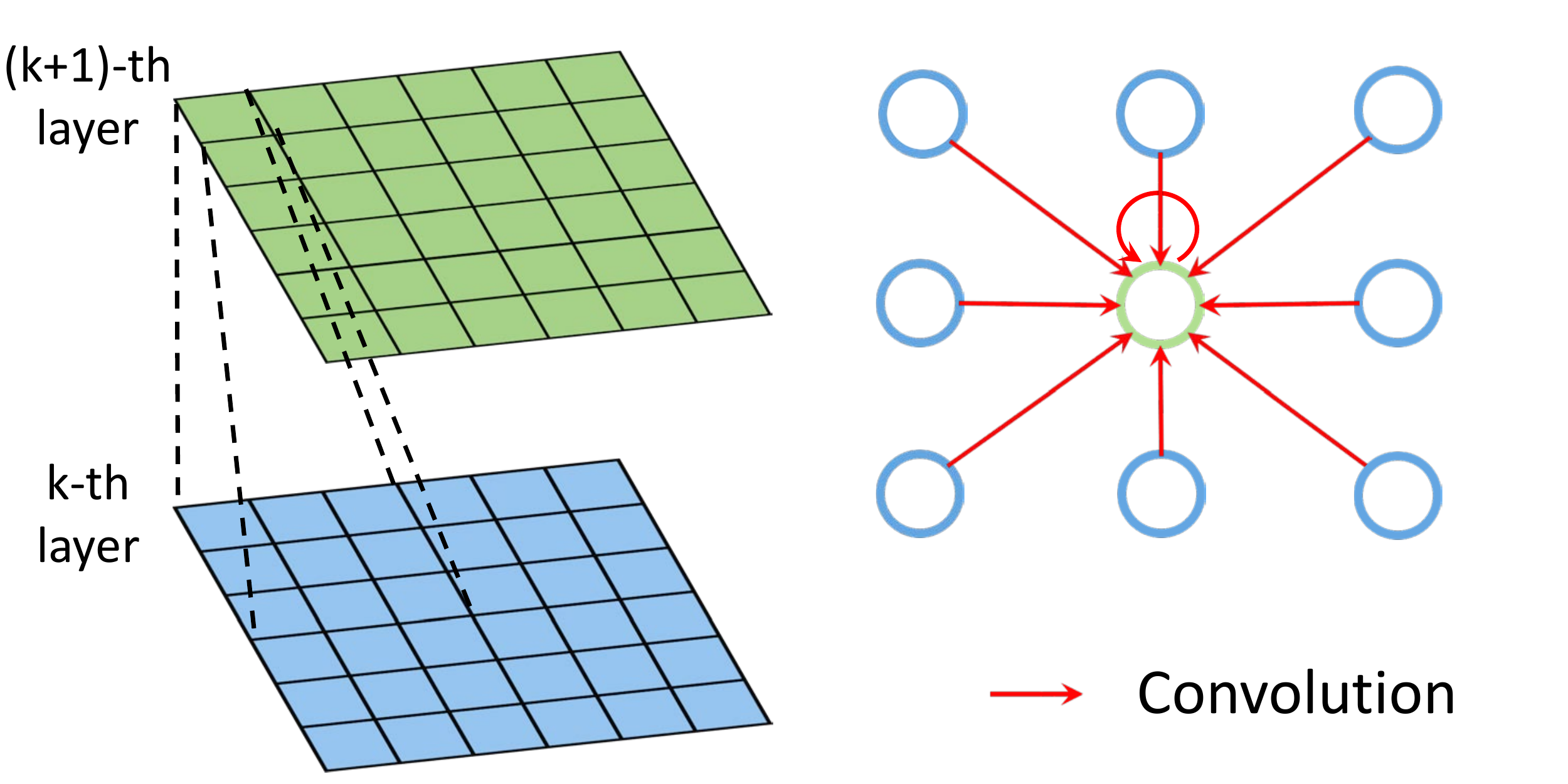}

		\end{minipage}%
	}%
	\hspace{.1in}
	\subfigure[An illustration of SGNNs \cite{xu2018powerful}. In each layer, each node aggregates from neighbor nodes and combines itself's hidden state.]{
		\begin{minipage}[t]{0.47\linewidth}
			\centering
			\includegraphics[width=1\linewidth]{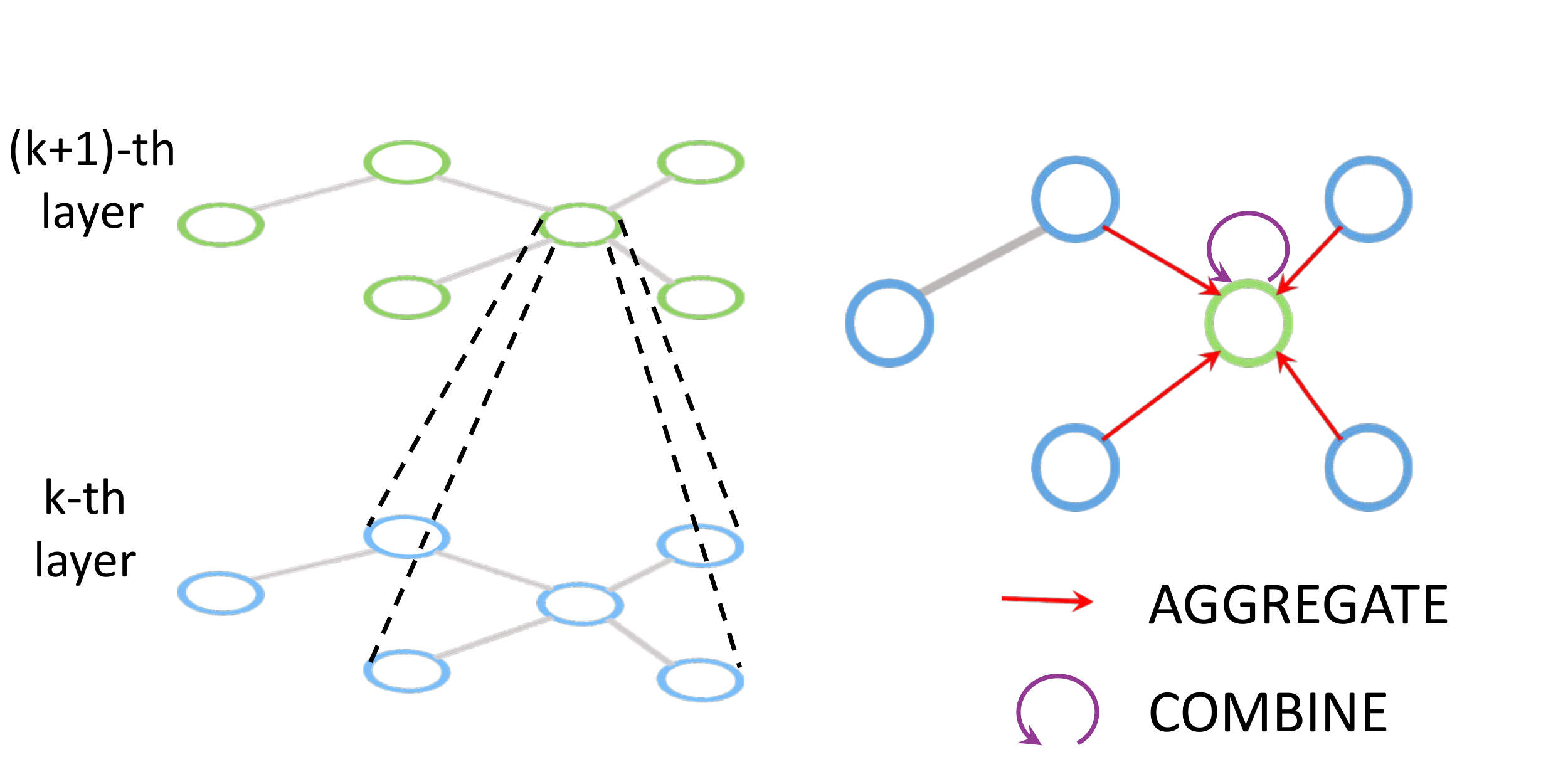}

		\end{minipage}%
	
	}%
	\centering
	
	\caption{Illustrations of CNN and SGNNs. CNN can be viewed as a special SGNNs on the grid graph.}
	\label{fig:CNNGNN}
\end{figure*}

In this subsection, we shall identify a family of GNNs for radio resource management problems, which extend CNNs to wireless channel graphs. In traditional machine learning tasks, the data can typically be embedded in a Euclidean space, e.g., images. Recently, there is an increasing number of applications generated from the non-Euclidean spaces that can be naturally modeled as graphs, e.g., point cloud \cite{wang2019dynamic} and combinatorial problems \cite{li2018combinatorial}. This motivates researchers to develop GNNs \cite{wu2019comprehensive}, which effectively exploit the graph structure. GNNs generalize traditional CNNs, recurrent neural networks, and auto-encoders to the graph tasks. In wireless networks, while the agents are located in the Euclidean space, channel states cannot be embedded in a Euclidean space. Thus, the data in radio resource management problems is also non-Euclidean and neural networks operating on non-Euclidean space are necessary when adopting ``learning to optimize'' approaches in wireless networks.

As a background, we first introduce CNNs, which operate on Euclidean data. Compared with MLPs, CNNs have shown superior performance in image processing tasks. The motivation for CNNs is that adjacent pixels are meaningful to be considered together in images \cite{brendel2019approximating}. Like MLPs, CNNs have a layer-wise structure. In each layer, a 2D convolution is applied to the input. Here we consider a simple CNN with a rectified linear unit and without pooling. In the $k$-th layer, for a pixel located at $(i,j)$, the update is 
\begin{equation}\label{eq:cnn}
\begin{aligned}
\bm{x}_{i,j}^{(k)} = \text{RELU}\left( \sum_{(p,l) \in \mathcal{N}(i,j)} \bm{W}_{i-p,j-l}^{(k)}\bm{x}_{p,l}^{(k-1)}   \right),
\end{aligned}
\end{equation}
where $\bm{x}_{i,j}^{(0)}$ denotes pixel $(i,j)$ of the input image, $\bm{x}_{i,j}^{(k)}$ denotes the hidden state of pixel $(i,j)$ at the $k$-th layer, $\text{RELU}(x) = \text{MAX}(0,x)$ and $\bm{W}^{(k)}_{\cdot,\cdot}$ denotes the weight matrix in the $k$-th layer, and $\mathcal{N}(i,j)$ denotes the neighbor pixels of pixel $(i,j)$. Specifically, for a convolution kernel of size $N\times N$, we have 
$$\mathcal{N}(i,j) = \left\{(p,l):|p-i| \leq \frac{N-1}{2}, |l-j| \leq \frac{N-1}{2} \right\}, $$
and a common choice of $N$ is $3$.

Despite the great success of CNNs in computer vision, they cannot be applied to non-Euclidean data. In \cite{xu2018powerful}, CNNs are extended to graphs from a spatial perspective, which is as efficient as CNNs, while enjoying performance guarantees on graph isomorphism test. We refer to this architecture as the \emph{spatial graph convolutional networks} (SGNNs). In each layer of a CNN (\ref{eq:cnn}), each pixel aggregates information from neighbor pixels and then updates its state. As an analogy, in each layer of a SGNN, each node updates its representation by aggregating features from its neighbor nodes. Specifically, the update rule of the $k$-th layer at vertex $i$ in a SGNN is 
\begin{equation} \label{gnns:SGNNs}
\begin{aligned}
  \bm{x}_i^{(k)} = \alpha^{(k)}\left(\bm{x}_i^{(k-1)}, \phi^{(k)} \left(\left\{\bm{x}_j^{(k-1)}: j \in \mathcal{N}(i)   \right\} \right)   \right),
\end{aligned}
\end{equation}
where $\bm{x}^{(0)}_i = \bm{Z}_{(i,:)}$ is the input feature of node $i$, $\bm{x}^{(k)}_i$ denotes the hidden state of node $i$ at the $k$-th layer, $\mathcal{N}(i)$ denotes the set of the neighbors of $i$, $\phi^{(k)}(\cdot)$ is a set function that aggregates information from the node's neighbors, and $\alpha^{(k)}(\cdot)$ is a function that combines aggregated information with its own information. An illustration of the extension from CNNs to SGNNs is shown in Fig. \ref{fig:CNNGNN}. Particularly, SGNNs include spatial deep learning for wireless scheduling \cite{cui2018spatial} as a special case.

Despite the success of SGNNs in graph problems, it is difficult to directly apply SGNNs on radio resource allocation problems as they cannot exploit the edge features.
This means that they cannot incorporate channel states in wireless networks. We modify the definition in (\ref{gnns:SGNNs}) to exploit edge features and will refer  to it as \emph{message passing graph neural networks} (MPGNNs). The update rule for the $k$-th layer at vertex $i$ in an MPGNN is  
\begin{equation}\label{eq:mpgnn}
\begin{aligned}
 \bm{x}_i^{(k)} = \alpha^{(k)}\left(\bm{x}_i^{(k-1)}, \phi^{(k)} \left(\left\{\left[\bm{x}_j^{(k-1)},\bm{e}_{j,i}\right]: j \in \mathcal{N}(i)   \right\} \right)  \right),
\end{aligned}
\end{equation}
where $\bm{e}_{j,i} = A_{(j,i,:)}$ is the edge feature of the edge $(j,i)$. We represent the output of a $S$-layer MPGNN as

\begin{align} \label{eq:output}
	\bm{X} =  \left[\bm{x}^{(S)}_1, \cdots, \bm{x}^{(S)}_{|V|}\right]^T \in \mathbb{C}^{|V| \times n} 
\end{align}

The extension from SGNNs to MPGNNs is simple but crucial, due to the following two reasons. First, MPGNNs respect the permutation equivariance property in Proposition \ref{prop:pe_cg}. Second, MPGNNs enjoy theoretical guarantees in radio resource management problems (as discussed in Section \ref{sec:analysis}). These two properties are unique for MPGNNs and are not enjoyed by SGNNs.

\subsection{Key Properties of MPGNNs} \label{sec:property}
MPGNNs enjoy properties that are favorable to solving large-scale radio resource management problems, as discussed in the sequel.

\paragraph{Permutation equivariance} 
We first show that MPGNNs satisfy the permutation equivariance property.

\begin{prop} (Permutation equivariance in MPGNNs) \label{prop:pe_gnn}
	Viewing the input output mapping of MPGNNs defined in (\ref{eq:mpgnn}) as $\Phi:(\bm{Z},A) \mapsto \bm{X}$ where $\bm{Z}$ is the node feature matrix, $A$ is the adjacency feature tensor and $\bm{X}$ is the output of an MPGNN in (\ref{eq:output}), we have
	\begin{align*}
	\Phi( (\pi \star \bm{Z}, \pi \star A)  ) = \pi \star \Phi((\bm{Z},A)),
	\end{align*}
	for any permutation operator $\pi$.
\end{prop}

\begin{rem}
 Please refer to Appendix \ref{proof:pe_gnn} for a detailed proof.
\end{rem}

The permutation equivariance property of GNNs improves the generalization of the neural networks. It also reduces the training sample complexity and training time. As shown in Proposition \ref{prop:pe_cg}, the radio resource management problems enjoy a permutation equivariant property. This means that the near-optimal solutions to a permuted problem are permutations of those to the original problem. GNNs well respect this property while MLPs and CNNs do not. If GNNs can perform well with a specific input, the good generalization is guaranteed with the permutation of this input, which is not guaranteed by MLPs or CNNs. Thus, in radio resource management problems, GNNs enjoy a better generalization than MLPs and CNNs. In contrast, to keep the input-output mapping in MLPs or CNNs permutation equivariant, data argumentation is needed. In principle, for each training sample, all its permutations should be put into the training dataset. This leads to a higher training sample complexity and more training time for MLPs and CNNs.

\paragraph{Ability to generalize to different problem scales} In MLPs, the input or output size must be the same during training and testing. Hence, the number of agents in the test dataset must be equal or less than that in the training dataset \cite{sun2018learning}. This means that MLP based methods cannot be directly applied to a different problem size. In MPGNNs, each node has a copy of two sub neural networks, i.e., $\alpha^{(k)}(\cdot)$ and $\phi^{(k)}(\cdot)$, whose input-output dimensions are invariant with the number of agents. Thus, we can train MPGNNs on small-scale problems and apply them to large-scale problems.

\paragraph{Fewer training samples} The required number of training samples for MPGNNs is much smaller than that for MLPs. The first reason is training sample reusing.
Note that the neural network at each node is identical. For each training sample, each node receives a permuted version of it and processes it with $\alpha^{(k)}(\cdot)$ and $\phi^{(k)}(\cdot)$.  Thus, each training sample is reused $L$ times for training $\{\alpha^{(k)}\}$ and $\{\phi^{(k)}\}$, where $L$ is the number of nodes. Second, input and output dimensions of the aggregation and combination functions in MPGNNs are much smaller than the original problem, which allows the use of much fewer parameters in neural networks. 

\paragraph{High computational efficiency} In each layer, an aggregation function is applied to all the edges and a combination function is applied to all the nodes. Thus, the time complexity for each layer is $\mathcal{O}(|E|+|V|)$ and the overall time complexity for an $L$-layer MPGNN is $\mathcal{O}(L(|E|+|V|))$. The time complexity grows \emph{linearly} with the number of agents when the maximal degree of the graph is bounded. Note that in MPGNNs, the aggregation function and combination function on each node can be executed in parallel. When the MPGNNs are fully parallelized, e.g., on powerful GPUs, the time complexity is $\mathcal{O}\left(LD\right)$, where $D$ is the maximal degree of the graph. This is a \emph{constant} time complexity when the maximal degree of the graph is bounded. We will verify this observation via simulations in Fig. \ref{fig:time}.

\subsection{An Effective Implementation of MPGNNs}

In this subsection, we propose an effective implementation of MPGNNs for radio resource management problems, named, the \emph{wireless channel graph convolution network} (WCGCN), which is able to effectively incorporate both agent-related features and channel-related features. The design space for MPGNNs (\ref{eq:mpgnn}) is to choose the set aggregation function $\phi(\cdot)$ and the combination function $\alpha(\cdot)$.

As general set functions are difficult to implement, an efficient implementation of $\phi(\cdot)$ was proposed in \cite{fey2019fast}, which has the following form
\begin{align*}
	\phi(\{x_1,\cdots, x_n\}) = \psi (\{ h(x_1), \cdots, h(x_n) \}),
\end{align*}
where $x_i, 1 \leq i \leq n$ are the elements in the set, $ \psi(\cdot)$ is a simple function, e.g., max or sum, and $h(\cdot)$ is some existing neural network architecture, e.g., linear mappings or MLPs. For $\alpha(\cdot)$ and $h(\cdot)$, linear mapping is adopted in popular GNN architectures (e.g., GCN \cite{kipf2016semi} and S2V \cite{dai2016discriminative}). Nevertheless, as discussed in Section IV in \cite{lee2019graph}, linear mappings have difficulty handling continuous features, which is ubiquitous in wireless networks (e.g., CSI). We adopt MLPs as $\alpha(\cdot)$ and $h(\cdot)$ for their approximation ability \cite{hornik1989multilayer}. MLP processing unit enables WCGCN to exploit complicated agent-related features and channel-related features in wireless networks.

For the aggregation function $\psi(\cdot)$, we notice that the following property holds if we use $\psi(\cdot) = \text{MAX}(\cdot)$.
 
\begin{thm} (Robustness to feature corruptions) \cite{qi2017pointnet} \label{thm:robust}
	Suppose $\bm{u}:\mathcal{X} \rightarrow \mathbb{R}^p$ such that $\bm{u}=\text{MAX}_{x_i \in S }$ and $f=\gamma \circ \bm{u}$. Then, 
	
	(a) $\forall S, \exists \mathcal{C}_S, \mathcal{N}_S \in \mathcal{X}, f(T) = f(S) \text{ if } \mathcal{C}_S \subset T \subset \mathcal{N}_S$;
	
	(b) $|\mathcal{C}_S| \leq p$.
\end{thm}

Theorem \ref{thm:robust} states that $f(S)$ remains the same up to corruptions of the input if all the features in $\mathcal{C}_S$ are preserved and $\mathcal{C}_S$ only contains a limited number of features, which is smaller than $p$. By specifying it to problems in wireless networks, the output of a layer remains unchanged even when the CSI is heavily corrupted on some links. In other words, it is robust to missing CSI.

We next specify the architecture for the WCGCN, which aligns with traditional optimization algorithms. First, in traditional optimization algorithms, each iteration outputs an updated version of the optimization variables. In the WCGCN, each layer outputs an updated version of the optimization variables. Second, these algorithms are often time-invariant systems, e.g., gradient descent, WMMSE \cite{Shi2011An}, and FPlinQ \cite{shen2017fplinq}. Thus, we share weights among different layers of the WCGCN. The update of the $i$-th node in the $k$-th layer can be rewritten as
\begin{equation}\label{eq:cgcnet}
\begin{aligned}
\bm{y}_i^{(k)} &= \text{MLP2}\left(\bm{x}_i^{(k-1)},  \text{MAX}_{j \in \mathcal{N}(i)}\left\{\text{MLP1} \left(\bm{x}_j^{(k-1)},\bm{e}_{j,i}\right)\right\}    \right),\\
\bm{x}_i^{(k)} &= \beta\left(\bm{y}_i^{(k)}\right),
\end{aligned}
\end{equation}
where $\text{MLP1}$ and $\text{MLP2}$ are two different MLPs, $\beta$ is a differentiable normalization function that depends on applications, $\bm{y}_i^{(k)}$ denotes the output of MLP2 of the $i$-th node in the $k$-th layer, $\bm{x}_i^{(k)}$ denotes the hidden state, and $\mathcal{N}(i)$ denotes the set of neighbor nodes of node $i$. For the power control problem, we constrain the power between $0$ and $1$, and $\beta$ can be a sigmoid function, i.e., $\beta(x) = \frac{1}{1+\exp(-x)}$. For more general constraints, $\beta$ can be differentiable projection layers \cite{agrawal2019differentiable}.

Besides the benign properties of MPGNNs, WCGCN enjoys several desirable properties for solving large-scale radio resource management problems. First, the WCGCN can effectively exploit features in multi-antenna systems with heterogeneous agents (e.g., channel states in multi-antenna systems and users' weights in weighted sum rate maximization). This is because WCGCN adopts MLP as processing units instead of linear mappings. This enables it to solve a wider class of radio resource management tasks than existing works \cite{cui2018spatial,lee2019graph,eisen2019optimal} (e.g., beamforming problems and weighted sum rate maximization). Second, it is robust to partial and imperfect CSI as suggested in Theorem \ref{thm:robust}.

\section{Theoretical Analysis of MPGNN-based Radio Resource Management} \label{sec:analysis}
In this section, we investigate performance and generalization of MPGNNs. We first prove the equivalence between MPGNNs and a family of distributed algorithms, which include many classic algorithms for radio resource management as special examples, e.g., WMMSE \cite{Shi2011An}. Based on this observation, we analyze the performance of MPGNN-based methods for weighted sum rate maximization problem.

\subsection{Simplifications}
To provide theoretical guarantees for ``learning to optimize'' approaches for solving radio resource management problems, it is critical to understand the performance and generalization of neural network-based methods. Unfortunately, the training and generalization of neural networks are sill open problems. We make several commonly adopted simplifications to make the performance analysis tractable. First, we focus on the MPGNN class instead of any specific neural network architecture such as GCNs.
Following Lemma 5 and Corollary 6 in \cite{xu2018powerful}, we can design an MPGNN with MLP processing units as powerful as the MPGNN class, and thus this simplification well serves our purpose. Second, we target at proving the existence of an MPGNN with performance guarantee. Because we train the neural network with a stochastic gradient descent with limited training samples during the simulations, we may not find the corresponding neural network parameters. While this may leave some gap between the theory and practice, our result is an important first step. These two simplifications have been commonly adopted in the performance analysis of GNNs \cite{xu2018powerful,sato2019approximation,barcelo2019logical}.

\subsection{Equivalence of MPGNNs and Distributed Optimization}
Compared with the neural network-based radio resource management, optimization-based radio resource management has been well studied. Thus, it is desirable to make connections between these two types of methods. In \cite{sato2019approximation}, the equivalence between SGNNs in (\ref{gnns:SGNNs}) and graph optimization algorithms was proved. Based on this result, we shall establish the equivalence between MPGNNs and a family of distributed radio resource management algorithms.

We first give a brief introduction to \emph{distributed local algorithms}, following \cite{hella2015weak}. The maximal degree of the nodes in the graph is assumed to be bounded. Distributed local algorithms are a family of iterative algorithms in a multi-agent system. In each iteration, each agent sends messages to its neighbors, receives messages from its neighbors, and updates its state based on the received messages. The algorithm terminates after a constant number of iterations. 

We focus on a sub-class of distributed local algorithms, titled, \emph{multiset broadcasting distributed local algorithms} (MB-DLA) \cite{hella2015weak}, which include a wide range of radio resource management algorithms in wireless networks, e.g., DTP \cite{kubisch2003distributed}, WMMSE \cite{Shi2011An}, FPlinQ \cite{shen2017fplinq}, and first-order methods for network utility problems \cite{shi2015large}.  Multiset and broadcasting refer to the way for receiving and sending messages, respectively. Denote $\bm{x}^{(l)}_i$ as the state of node $i$ at the $l$-th iteration, and the MB-DLA is shown in Algorithm \ref{alg:mb-dla}.

\begin{algorithm}  
	\footnotesize
	\caption{Multiset broadcasting distributed local algorithm \cite{hella2015weak}}  
	\label{alg:mb-dla}  
	\begin{algorithmic}[1]
		\State Initialize all internal states $\bm{x}_k^{(0)}, \forall k$.
		\For{communication round $t=1,\cdots,T$}
		\State agent $k$ sends $h^{(t)}_1(\bm{x}^{(t-1)}_k)$ to all its edges, $\forall k$
		\State agent $k$ receives $\left\{\bm{m}_{j,k}^{(t)}|\bm{m}_{j,k}^{(t)} = h_2^{(t)}\left(h^{(t)}_1(\bm{x}^{(t-1)}_k),  \bm{e}_{j,k}\right), j\in \mathcal{N}(k)\right\}$ from the edges, $\forall k$
		\State agent $k$ updates its internal state $\bm{x}^{(t)}_k = g_2^{(t)}\left(\bm{x}^{(t-1)}_k,  g_1^{(t)}\left(\left\{\bm{m}^{(t)}_{j,k}: j \in \mathcal{N}(k)\right\}\right)\right)$, $\forall k$.
		\EndFor
		\State Output $\bm{x}_k^{(T)}$.
	\end{algorithmic}  
\end{algorithm}

The equivalence between MPGNNs and MB-DLAs roots in the similarity in their definitions. In each iteration of an MB-DLA, each agent aggregates messages from neighbor agents and updates its local state. In each layer of an MPGNN, each node aggregates features from neighbor nodes. The equivalence can be drawn if we view the agents as nodes in a graph and messages as the features. The following proposition states the equivalence of MPGNNs and MB-DLAs formally.

\begin{thm}  \label{thm:MPGNNMB}
	Let MB-DLA($T$) denote the family of MB-DLA with $T$ iterations and MPGNN($T$) as the family of MPGNNs with $T$ layers, then the following two conclusions hold.
	\begin{enumerate}
		\item For any MPGNN($T$), there exists a distributed local algorithm in MB-DLA($T$) that solves the same set of problems as MPGNN($T$).
		\item For any algorithm in MB-DLA($T$), there exists an MPGNN($T$) that solves the same set of problems as this algorithm.
	\end{enumerate}
	
\end{thm}

\begin{rem}
	 Please refer to Appendix \ref{proof:MPGNNMB} for a detailed proof.
\end{rem}

The equivalence allows us to analyze the performance of MPGNNs by studying the performance of MB-DLAs. The first result shows that MPGNNs are \emph{at most} as powerful as MB-DLAs. The implication is that if we can prove that there is no MB-DLA capable of solving a specific radio resource management problem, then MPGNNs cannot solve it. This can be used to prove a performance upper bound of MPGNNs. The second result shows that MPGNNs are \emph{as powerful as} MB-DLAs in radio resource management problems. This implies that if we are able to identify an MB-DLA that solves a radio resource management problem well, then there exists an MPGNN performs better or at least competitive. The generalization is also as good as the corresponding MB-DLA. We shall give a specific example on sum rate maximization in the next subsection.

\subsection{Performance and Generalization of MPGNNs} \label{sec:per}
In this subsection, we use the tools developed in the last subsection to analyze the performance and generalization of MPGNNs in the sum rate maximization problem. The analysis is built on the observation that a classic algorithm for the sum rate maximization problem, i.e., WMMSE, is an MB-DLA under some conditions, which is formally stated below. We shall refer to the MB-DLA corresponding to WMMSE as WMMSE-DLA.

\begin{prop} \label{prop:wmmse}
	When the maximal number of interference neighbors is bounded by some constant, then WMMSE with a constant number of iterations is an MB-DLA. 
\end{prop}

\begin{rem}
	When the problem sizes in the training dataset and test dataset are the same, we can always assume that the number of interference neighbors is a common constant. The restriction of a constant number of interference neighbors only influences the generalization. Please refer to Appendix \ref{proof:wmmse} for a detailed proof.
\end{rem}

\paragraph{Performance} Proposition \ref{prop:wmmse} shows that WMMSE is an MB-DLA. Thus, when the problem sizes in the training dataset and test dataset are the same, there exists an MPGNN whose performance is as good as WMMSE. As the WMMSE is hand-crafted, it is not optimal in terms of the number of iterations. By employing a unsupervised loss function, we expect that MPGNNs can learn an algorithm which has fewer iterations and may possibly enjoy better performance. In Fig. \ref{fig:wmmse}, we observe that a $1$-layer MPGNN outperforms WMMSE with $10$ iterations and a $2$-layer MPGNN outperforms WMMSE with $30$ iterations.

\paragraph{Generalization} To avoid the excessive training cost, it is desirable to first train a neural network on small-scale problems and then generalize it to large-scale ones. An intriguing question is when such generalization is reliable. Compared with WMMSE, WMMSE-DLA has two constraints: Both the number of iterations and the maximal number of interference neighbors should be bounded by some constants. As agents that are far away cause little interference, the number of interference neighbors can be assumed to be fixed when the user density is kept the same. As a result, the performance of MPGNNs is stable when the user density in the test dataset is the user density in the training dataset multiplied by a constant. We will verify this by simulations in Table \ref{tab:wp_density} and Table \ref{tab:bf_density}.

\section{Simulation Results} \label{sec:exp}
In this section, we provide simulation results to verify the effectiveness of the proposed neural network architecture for three applications. The first application is sum rate maximization in a Gaussian interference channel, which is a classic application for deep learning-based methods. We use this application to compare the proposed method with MLP-based methods \cite{liang2018towards} and optimization-based methods \cite{Shi2011An}. The second application is weighted sum rate maximization, and the third application is beamformer design. The last two problems cannot be solved by existing methods in \cite{cui2018spatial,lee2019graph,eisen2019optimal}. 

For the neural network setting, we adopt a $3$-layer WCGCN, implemented by Pytorch Geometric \cite{fey2019fast}. During the training, the neural network takes channel states and users' weights as input and outputs the beamforming vector for each user. We apply the following loss function at the last layer of the neural network. 

\begin{align*}
\ell(\bm{\Theta}) = -\mathbb{E}\left(\sum_{k=1}^{K} w_{k} \log_2 \left(1+ \frac{|\bm{h}_{k,k}^H \bm{v}_{k}(\bm{\Theta})|^2}{ \sum_{j\neq k}^K |\bm{h}_{j,k}^H\bm{v}_{j}(\bm{\Theta})|^2 + \sigma_{k}^2} \right)\right),
\end{align*}
where $\bm{\Theta}$ denotes the weights of the neural network and the expectation is taken over all the channel realizations. By adopting this loss function, no labels are required and thus it is an unsupervised learning method. In the training stage, to optimize the neural network, we adopt the adam optimizer \cite{kingma2014adam} with a learning rate of $0.001$. In the test stage, the input of the neural network consists of the channel states and users' weights and the output of the neural network is the beamforming vector. The SGD (adam) optimizer is not needed in the test stage.

\subsection{Sum Rate Maximization}
We first consider the sum rate maximization problem in a single-antenna Gaussian interference channel. This problem is a special case of (\ref{eq:sys_mod}) with $N_t = 1$, $h_{j,k} \sim \mathcal{CN}(0,1), \forall j,k$, and $w_k=1,\forall k$.

We consider the following benchmarks for comparison.
\begin{itemize}
	\item \textbf{WMMSE} \cite{Shi2011An}: This is a classic optimization-based algorithm for sum utility maximization in MIMO interfering broadcast channels. We run WMMSE for $100$ iterations with random initialization.
	\item \textbf{WMMSE 100 times}: For each channel realization, we run WMMSE algorithm for $100$ times and take the best one as the performance. This is often used as an performance upper bound.
	\item \textbf{Strongest}: We find a fixed proportion of pairs with the largest channel gain $|h_{i,i}|$, and set the power of these pairs as $P_{\text{max}}$ while the power levels for remaining pairs are set to $0$. This is a simple baseline algorithm without any knowledge of interference links.
	\item \textbf{PCNet} \cite{liang2018towards}: PCNet is an MLP based method particularly designed for the sum rate maximization problem with single-antenna channels.
\end{itemize}

We use $10^{4}$ training samples for WCGCN and $10^{7}$ training samples for PCNet. For a specific parameter setting of WCGCN (\ref{eq:cgcnet}), we set the hidden units of MLP1 in (\ref{eq:cgcnet}) as $\{5,32,32\}$, MLP2 as $\{35,16,1\}$, and $\beta(\cdot)$ as sigmoid function.\footnote{The performance of WCGCN is not sensitive to the number of hidden units.} The performance of different methods is shown in Table \ref{tab:uw_g}. The SNR and number of users are kept the same in the training and test dataset. For all the tables shown in this section, the entries are (weighted) the sum rates achieved by different methods normalized by the sum rate of WMMSE. We see that both PCNet and WCGCN achieve near-optimal performance when the problem scale is small. As the problem scale becomes large, the performance of PCNet approaches Strongest. This shows that it can hardly learn any valuable information about interference links. Nevertheless, the performance of WCGCN is stable as the problem size increases. Thus, GNNs are more favorable than MLPs for medium-scale or large-scale problems.

\begin{table}[htb]
	
	\selectfont  
	\centering
	
	\caption{Average sum rate under each setting. The results are normalized by the sum rate of WMMSE.} 
	\newcommand{\tabincell}[2]{\begin{tabular}{@{}#1@{}}#2\end{tabular}}
	\resizebox{0.48\textwidth}{!}{
		\begin{tabular}{|c|c|c|c|c|c|}
			\hline
			SNR& Links  & WCGCN & PCNet  & Strongest & \tabincell{c}{WMMSE \\$100$ times}\cr \hline
			\multirow{3}{*}{0dB} & 10  &$100.0\%$ & $98.9\%$  & $87.1\%$ & $102.0\%$ \cr \cline{2-6} 
			& 30 & $97.9\%$ & $87.4\%$ & $82.8\%$&$101.3\%$ \cr \cline{2-6} 
			& 50 & $97.1\%$ & $79.7\%$  & $80.6\%$&$101.4\%$ \cr \hline
			\multirow{3}{*}{10dB} & 10 & $103.1\%$ & $101.8\%$ & $74.4\%$ & $104.8\%$ \cr \cline{2-6} 
			& 30 & $103.4\%$ & $74.0\%$ & $70.0\%$& $105.0\%$ \cr \cline{2-6} 
			& 50 & $102.5\%$ & $67.0\%$ & $68.9\%$&$104.7\%$ \cr \hline
	\end{tabular}}
	\label{tab:uw_g}
\end{table}

We further compare the performance of WCGCN and WMMSE with different numbers of iterations. We use the system setting $K=50$, $\text{SNR}=10\text{dB}$. Both WMMSE and WCGCN starts from the same initialization point. The results are shown in Fig. \ref{fig:wmmse}. From the figure, we see that a $1$-layer WCGCN outperforms WMMSE with $10$ iterations and a $2$-layer WCGCN outperforms WMMSE with $30$ iterations. This indicates that by adopting the unsupervised loss function, WCGCN can learn a much better message-passing algorithm than the handcrafted WMMSE.

\begin{figure}[htb]
	\centering
	\includegraphics[width=0.5\textwidth]{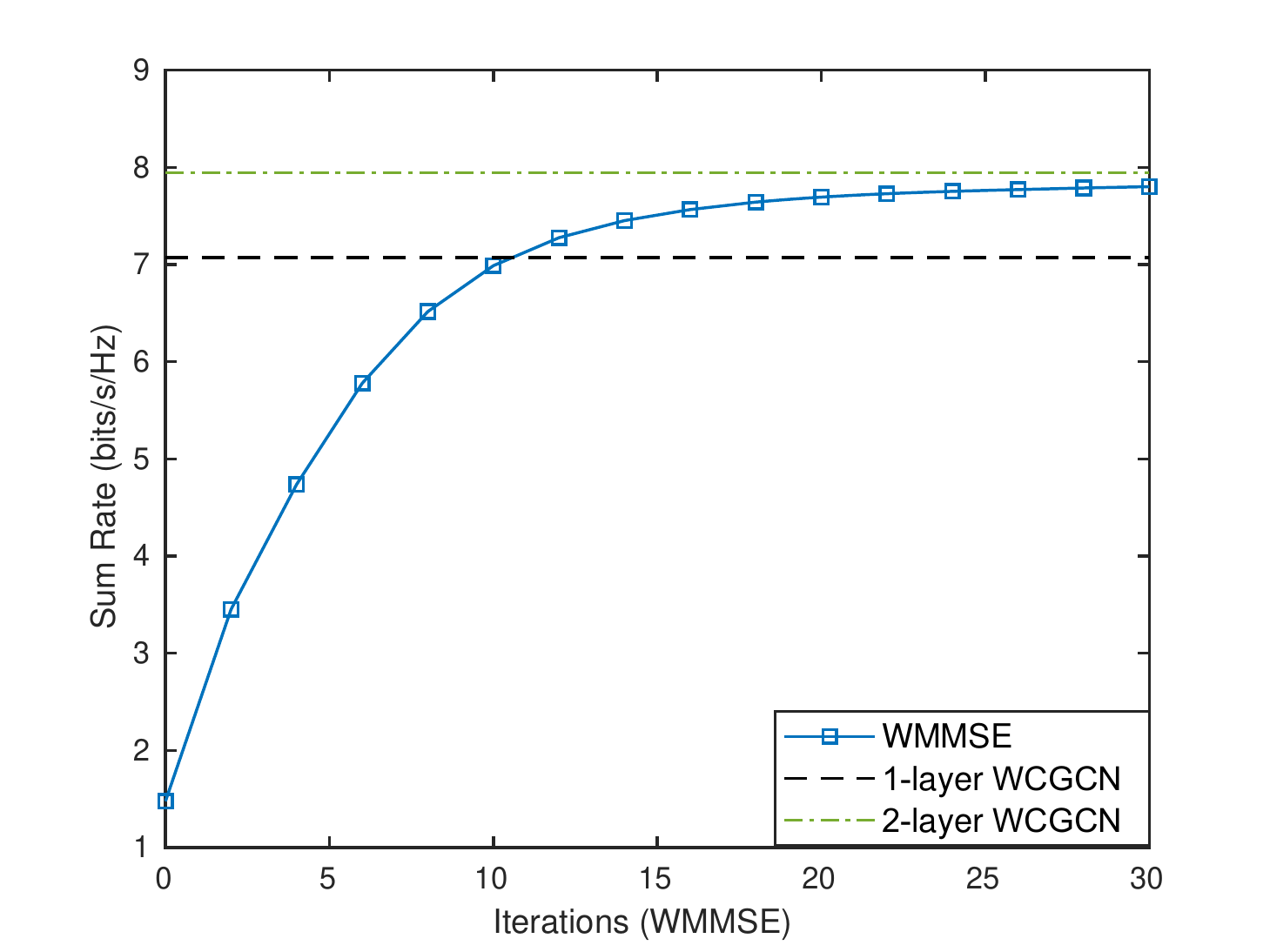}
	\caption{A comparison between WCGCN and WMMSE with different numbers of iterations.}
	\label{fig:wmmse}
\end{figure}

\subsection{Weighted Sum Rate Maximization}
In this application, we consider $K$ single-antenna transceiver pairs  within a $A \times A$ area. The transmitters are randomly located in the $A \times A$ area while each receiver is uniformly distributed within $[d_{\text{min}},d_{\text{max}}]$ from the corresponding transmitter. We adopt the channel model from \cite{shi2014group} and use $10000$ training samples for each setting. To reduce the CSI training overhead, we assume $h_{j,k}$ is available to WCGCN only if the distance between transmitter $j$ and receiver $k$ is within $500$ meters. To provide a performance upper bound, global CSI is assumed to be available to WMMSE. The weights for weighted sum rate maximization, i.e., $w_k$ in (\ref{eq:sys_mod}), are generated from a uniform distribution in $[0,1]$ in both training and test dataset. For a specific parameter setting of WCGCN (\ref{eq:cgcnet}), we set the hidden units of MLP1 as $\{5,32,32\}$, MLP2 as $\{35,16,1\}$, and $\beta(\cdot)$ as sigmoid function.

\paragraph{Performance comparison} We first test the performance of WCGCN when the number of pairs is the same in the training and test dataset. Specifically, we consider $K=50$ pairs in a $1000\text{m} \times 1000\text{m}$ region. We test the performance of WCGCN with different values of $d_{\text{min}}$ and $d_{\text{max}}$, as shown in Table \ref{tab:wp_same}. The entries in the table are the sum rates achieved by different methods. We observe that WCGCN with local CSI achieves competitive performance to WMMSE with global CSI.

\begin{table}[htb]
	
	\selectfont  
	\centering
	
	\caption{Average sum rate performance of $50$ transceiver pairs.} 
	
	\resizebox{0.48\textwidth}{!}{
		\begin{tabular}{|c|c|c|c|c|c|c|}  
			\hline  
			$(d_{\text{min}}, d_{\text{max}})$&   (2m,65m)        &    (10m,50m)       &     (30m,70m)      &      (30m,30m)     \\ \hline
			WCGCN &      $97.8\%$     &      $97.5\%$      &       $96.5\%$      &      $96.8\%$     \\ \hline

	\end{tabular}}
	\label{tab:wp_same}
\end{table}

Next, to test the generalization capability of the proposed method, we train WCGCN on a wireless network with tens of users and test it on wireless networks with hundreds or thousands of users, as shown in the following two simulations.

\paragraph{Generalization to larger scales}   We first train the WCGCN with $50$ pairs in a $1000\text{m} \times 1000\text{m}$ region. We then change the number of pairs in the test set while the density of users (i.e.,  $A^2/K$)  is fixed. The results are shown in Table \ref{tab:wp_scale}. It can be observed that the performance is stable as the number of users increases. It also  shows that WCGCN can well generalize to larger problem scales, which is consistent with our analysis. 

\begin{table}[htb]
	
	\selectfont  
	\centering
	
	\caption{Generalization to larger problem scales but same density.} 
	
	\resizebox{0.4\textwidth}{!}{
		\begin{tabular}{|c|c|c|c|}
			\hline
		Links	& Size ($m^2$) & \multicolumn{2}{c|}{$(d_{\text{min} },d_\text{max})$} \\ \hline
 	&  &        (10m,50m)     &    (30m,30m)       \\ \hline
		$200$	& $2000 \times 2000$ &       $98.3\%$    &   $98.1\%$        \\ \hline
		$400$	& $2828 \times 2828 $  &     $98.9\%$      &   $98.2\%$        \\ \hline
		$600$	& $3464 \times 3464$ &      $98.8\%$      &     $98.7\%$        \\ \hline
		$800$	& $4000 \times 4000$ &      $98.9\%$     &     $98.6\%$      \\ \hline
		$1000$	& $4772 \times 4772$  &   $98.9\%$        &   $98.7\%$        \\ \hline
	\end{tabular}}
	\label{tab:wp_scale}
\end{table}

\paragraph{Generalization to higher densities}  In this test, we first train the WCGCN with $50$ pairs in a $1000\text{m} \times 1000\text{m}$ region. We then change the number of pairs in the test set while fixing the area size. The results are shown in Table \ref{tab:wp_density} and the performance loss compared with $K=50$ is shown in the bracket. The performance is stable up to a $4$-fold increase in the density, and good performance is achieved even when there is a $10$-fold increase in the density.

\begin{table}[htb]
	
	\selectfont  
	\centering
	
	\caption{Generalization over different link densities. The performance loss compared to $K=50$ is shown in the bracket.} 
	
	\resizebox{0.48\textwidth}{!}{
		\begin{tabular}{|c|c|c|c|}
			\hline
			Links	& Size ($m^2$) & \multicolumn{2}{c|}{$(d_{\text{min} },d_\text{max})$} \\ \hline
			&  &        (10m,50m)     &    (30m,30m)       \\ \hline
			$100$& \multirow{5}{*}{$1000 \times 1000$}  &  $97.6\%$ ($+0.1\%$)         &     $96.4\%$ ($-0.1\%$)      \\ \cline{1-1} \cline{3-4} 
			$200$&                    &     $97.0\%$ ($-0.5\%$)      &      $96.0\%$ ($-0.5\%$)     \\ \cline{1-1} \cline{3-4} 
			$300$&                    &     $95.9\%$ ($-1.6\%$)      &      $94.9\%$ ($-1.6\%$)     \\ \cline{1-1} \cline{3-4} 
			$400$&                    &     $95.6\%$  ($-1.9\%$)      &     $94.5\%$ ($-2.0\%$)      \\ \cline{1-1} \cline{3-4} 
			$500$&                    &      $95.3\%$ ($-2.2\%$)     &      $94.5\%$ ($-2.0\%$)     \\ \hline
	\end{tabular}}
	\label{tab:wp_density}
\end{table}

\subsection{Beamformer Design} \label{sec:ex_bf}
In this subsection, we consider the beamforming for sum rate maximization in (\ref{eq:sys_mod}). Specifically, we consider $K$ transceiver pairs within a $A \times A$ area, where the transmitters are equipped with multiple antennas and each receiver is equipped with a single antenna. The transmitters are generated uniformly in the area and the receivers are generated uniformly within $[d_{\text{min}},d_{\text{max}}]$ from the corresponding transmitters. We adopt the channel model in \cite{shi2014group} and use $50000$ training samples for each setting. The assumption of the available CSI for WCGCN and WMMSE is the same as the previous subsection. In WCGCN, a complex number is treated as two real numbers. For a specific parameter setting of WCGCN (\ref{eq:cgcnet}), we set the hidden units of MLP1 as $\{6N_t, 64, 64\}$, MLP2 as $\{64+4N_t, 32, 2N_t\}$, and $\beta(\bm{x}) = \frac{\bm{x}}{\text{max}(\|\bm{x}\|_2,1)}$. 

\paragraph{Performance comparison} We first test the performance of WCGCN when the number of pairs in the training dataset and the number of pairs in the test dataset are the same. Specifically, we consider $K=50$ pairs in a $1000$ meters by $1000$ meters region and each transmitter is equipped with $2$ antennas. We test the performance of WCGCN with different $d_{\text{min}}$ and $d_{\text{max}}$. The results are shown in Table \ref{tab:bf_same}. We observe that WCGCN achieves comparable performance to WMMSE with local CSI, demonstrating the applicability of the proposed method to multi-antenna systems.

\begin{table}[htb]
	
	\selectfont  
	\centering
	
	\caption{Average sum rate performance of $50$ transceiver pairs with $N_t = 2$. The results are normalized by the sum rate of WMMSE.} 
	
	\resizebox{0.48\textwidth}{!}{
		\begin{tabular}{|c|c|c|c|c|c|c|}  
			\hline  
			$(d_{\text{min}}, d_{\text{max}})$&   (2m,65m)        &    (10m,50m)       &     (30m,70m)      &      (30m,30m)     \\ \hline
			WCGCN &      $97.1\%$     &      $96.0\%$      &       $94.1\%$      &      $96.2\%$     \\ \hline

	\end{tabular}}
	\label{tab:bf_same}
\end{table}

\paragraph{Generalization to larger scales}  We first train the WCGCN with $50$ pairs in a $1000$ meters by $1000$ meters region with $N_t = 2$. We then change the number of pairs while the density of users (i.e.,  $A^2/K$)  is fixed. The results are shown in Table \ref{tab:bf_scale}. The performance is stable as the number of users increases, which is consistent with our theoretical analysis.

\begin{table}[htb]
	
	\selectfont  
	\centering
	
	\caption{Generalization to larger problem scales but same density.} 
	
	\resizebox{0.4\textwidth}{!}{
		\begin{tabular}{|c|c|c|c|}
			\hline
			Links	& Size ($m^2$) & \multicolumn{2}{c|}{$(d_{\text{min} },d_\text{max})$} \\ \hline
			&  &        (2m,65m)     &    (10m,50m)       \\ \hline
			$200$	& $2000 \times 2000$ &       $97.3\%$    &  $96.8\%$       \\ \hline
			$400$	& $2828 \times 2828 $  &     $97.3\%$      &   $96.7\%$        \\ \hline
			$600$	& $3464 \times 3464$ &      $97.2\%$      &    $96.5\%$        \\ \hline
			$800$	& $4000 \times 4000$ &      $97.2\%$    &    $96.5\%$     \\ \hline
			$1000$	& $4772 \times 4772$  &   $97.2\%$        &   $96.4\%$        \\ \hline
	\end{tabular}}
	\label{tab:bf_scale}
\end{table}

\paragraph{Generalization to larger densities}  We first train the WCGCN with $50$ pairs on a $1000$ meters by $1000$ meters region with $N_t = 2$. We then change the number of pairs while fix the area size. The results are shown in Table \ref{tab:bf_density} and the performance loss is shown in the bracket. The performance is stable up to a $2$-fold increase in the density and satisfactory performance is achieved up to a $4$-fold increase in the density.  The performance deteriorates when the density grows, which indicates that extra training is needed when the density in the test dataset is much larger than that of the training dataset.

\begin{table}[htb]
	
	\selectfont  
	\centering
	
	\caption{Generalization over different link densities. The performance loss compared to $K=50$ is shown in the bracket.} 
	
	\resizebox{0.48\textwidth}{!}{
		\begin{tabular}{|c|c|c|c|}
			\hline
			Links	& Size ($m^2$) & \multicolumn{2}{c|}{$(d_{\text{min} },d_\text{max})$} \\ \hline
			&  &        (10m,50m)     &    (30m,30m)       \\ \hline
			$100$& \multirow{5}{*}{$1000 \times 1000$}  &  $97.0\%$ ($-0.1\%$)         &    $95.7\%$ ($-0.3\%$)      \\ \cline{1-1} \cline{3-4} 
			$200$&                    &     $95.8\%$($-1.3\%$)      &      $94.4\%$ ($-1.6\%$)     \\ \cline{1-1} \cline{3-4} 
			$300$&                    &     $94.5\%$ ($-2.6\%$)      &      $93.0\%$ ($-3.0\%$)     \\ \cline{1-1} \cline{3-4} 
			$400$&                    &     $92.5\%$ ($-4.6\%$)      &     $92.0\%$ ($-4.0\%$)      \\ \cline{1-1} \cline{3-4} 
			$500$&                    &      $91.4\%$ ($-5.7\%$)     &      $90.7\%$($-5.3\%$)     \\ \hline
	\end{tabular}}
	\label{tab:bf_density}
\end{table}

\begin{figure}[htb]
	\centering
	\includegraphics[width=0.5\textwidth]{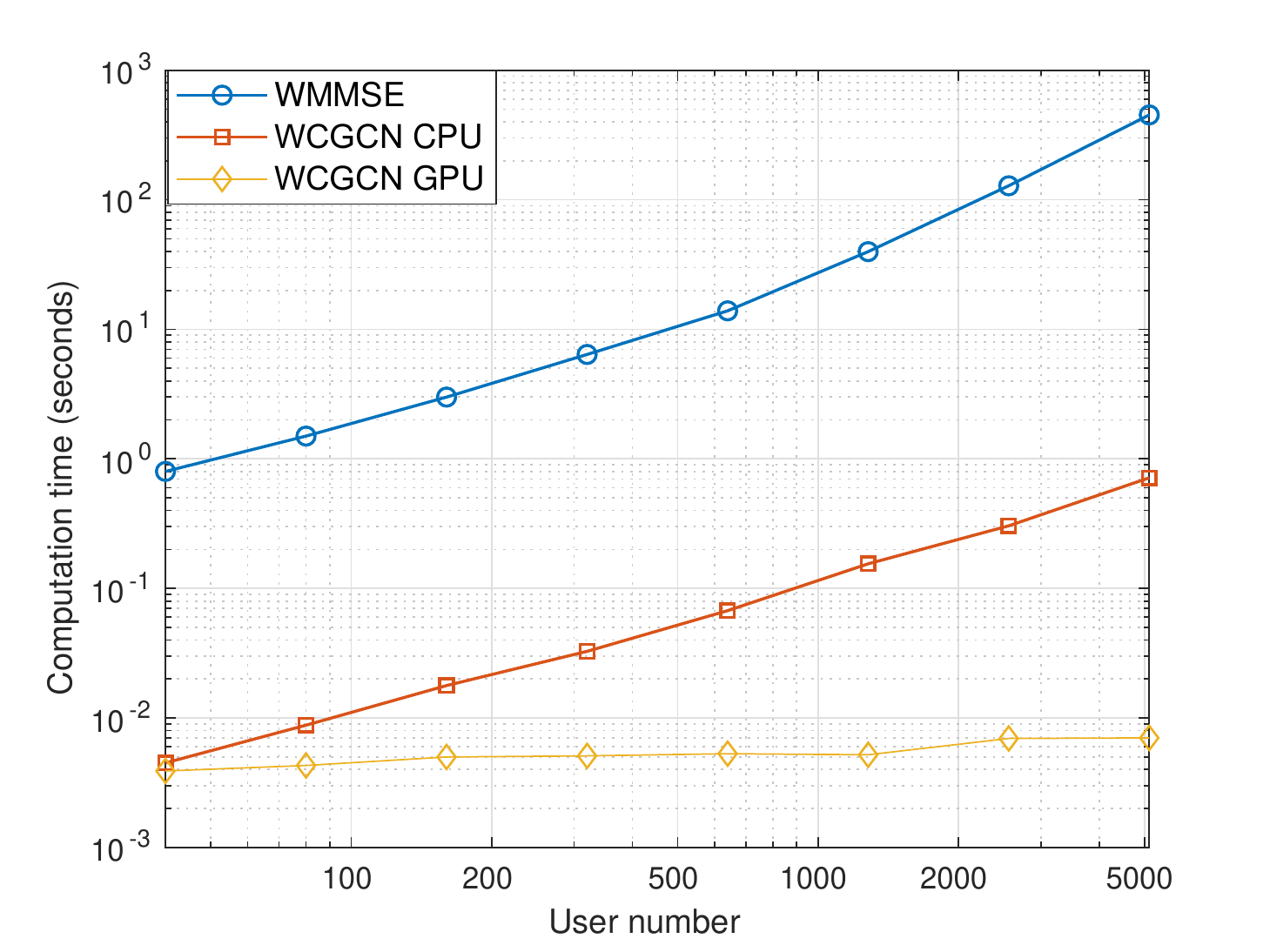}
	\caption{Computation time comparison of different methods.}
	\label{fig:time}
\end{figure}

\paragraph{Computation time comparison} This test compares the running time of different methods for different problem scales. We run ``WCGCN GPU'' on GeForce GTX 1080Ti while the other methods on Intel(R) Xeon(R) CPU E5-2643 v4 @ 3.40GHz. The implementation of neural networks exploits the parallel computation of GPU while WMMSE is not able to do so due to its sequential computation flows. The running time is averaged over $50$ problem instances and shown in Fig. \ref{fig:time}. The speedup compared with WMMSE becomes large as the problem scale increases. This benefits from the low computational complexity of WCGCN. As shown in the figure, the computational complexity of WCGCN CPU is linear and WCGCN GPU is nearly a constant, which is consistent with our analysis in Section \ref{sec:property}. Remarkably, WCGCN is able to solve the problem with $1000$ users within $6$ milliseconds.

\section{Conclusions}
In this paper, we developed a scalable neural network architecture based on GNNs to solve radio resource management problems. In contrast to existing learning based methods, we focused on the neural architecture design to meet the key performance requirements, including low training cost, high computational efficiency, and good generalization. Moreover, we theoretically connected learning based methods and optimization based methods, which casts light on the performance guarantee of learning to optimize approaches. We believe that this investigation will lead to profound implications in both theoretical and practical aspects. As for future directions, it will be interesting to investigate the distributed deployment of MPGNNs for radio resource management in wireless networks, and extend our theoretical results to more general application scenarios.

 \appendices 
 \section{Proof of Proposition \ref{prop:pe_dg}} \label{proof:pe_dg}
 Following Proposition \ref{prop:pi_dg}, we have
 \begin{align}\label{eq:pi}
 g(\bm{\Gamma},\bm{A}) = g(\pi \star \bm{\Gamma},\pi \star \bm{A}),  \quad Q(\bm{\Gamma},\bm{A}) = Q(\pi \star \bm{\Gamma},\pi \star \bm{A}),
 \end{align}
 for any variable $\bm{\Gamma}$, adjacency matrix $\bm{A}$, and permutation matrix $\bm{P}$.
 
 For any $\bm{\Gamma} \in \mathcal{R}^{\alpha}_g$, we have 
 \begin{align} \label{eq:sublevel}
 &g(\pi \star \bm{\Gamma},\pi \star \bm{A}) = g(\bm{\Gamma},\bm{A}) \leq \alpha, \notag \\ 
 & Q(\pi \star \bm{\Gamma},\pi \star \bm{A}) = Q(\bm{\Gamma},\bm{A}) \leq 0.
 \end{align}
 Combining (\ref{eq:pi}) and (\ref{eq:sublevel}), we have
 \begin{align*}
 &F(\pi \star \bm{A}) = \{\bm{\Gamma}|g(\bm{\Gamma},\pi \star \bm{A}) \leq \alpha, Q(\bm{\Gamma},\pi \star \bm{A})\leq 0  \} \\ 
 = &\{\pi \star \bm{\Gamma}|g(\pi \star \bm{\Gamma},\pi \star \bm{A}) \leq \alpha, Q(\pi \star \bm{\Gamma},\pi \star \bm{A})\leq 0  \} \\
 = &\{\pi \star \bm{\Gamma}|g(\bm{\Gamma},\bm{A}) \leq \alpha, Q(\bm{\Gamma},\bm{A})\leq 0  \} \\
 = &\{\pi \star \bm{\Gamma}|\bm{\Gamma} \in \mathcal{R}^{\alpha}_g\}.
 \end{align*}
 
 \vspace{-1em}
 \section{Proof of Proposition \ref{prop:pe_gnn}} \label{proof:pe_gnn}
 In the original graph, denote the input feature of node $i$ as $\bm{x}_i^{(0)}$, the edge feature of edge $(j,i)$ as $\bm{e}_{j,i}$, and the output of the $k$-th layer of node $i$ as $\bm{x}^{(k)}_i$. In the permuted graph, denote the input feature of node $i$ as $\hat{\bm{x}}_i^{(0)}$, the edge feature of edge $(j,i)$ as $\hat{\bm{e}}_{j,i}$, and the output of the $k$-th layer for node $i$ as $\hat{\bm{x}}^{(k)}_i$. Due to the permutation relationship, we have 
 \begin{equation}\label{pe:0}
 	\begin{aligned} 
 		&\hat{\bm{e}}_{\pi(j),\pi(i)} = \bm{e}_{j,i}, \quad \hat{\bm{x}}_{\pi(i)}^{(0)} = \bm{x}_i^{(0)}, \\ &\mathcal{N}(\pi(i) ) = \{\pi(j), j \in \mathcal{N}(i) \}.
 	\end{aligned} 
 \end{equation}

 For any fixed $\pi$ and $n$, we prove $\bm{x}_i^{(n)} = \hat{\bm{x}}_{\pi(i)}^{(n)}, \forall i$ by induction. 
 1) The base case of $n = 0$ follows (\ref{pe:0}).
 
 2) Assume $\bm{x}_i^{(s-1)} = \hat{\bm{x}}_{\pi(i)}^{(s-1)}, \forall i$ when $n = s - 1$.
 	Show $n = s$ holds: In the $s$-th layer, the following update rule is applied
 	\begin{equation}\label{upd:layer}
 	\begin{aligned} 
 	&\bm{x}_i^{(s)} = \alpha^{(s)}\left(\bm{x}_i^{(s-1)}, \phi^{(s)} \left\{\left[\bm{x}_j^{(s-1)},\bm{e}_{j,i}\right]: j \in \mathcal{N}(i)   \right\}   \right), \\
 	&\hat{\bm{x}}_{\pi(i)}^{(s)} = \alpha^{(s)}\left(\hat{\bm{x}}_{\pi(i)}^{(s-1)}, \phi^{(s)} \left\{\left[\hat{\bm{x}}_{j}^{(s-1)},\hat{\bm{e}}_{j,\pi(i)}\right]: j \in \mathcal{N}(\pi(i))   \right\}   \right).
 	\end{aligned}
 	\end{equation}

 Following (\ref{pe:0}), (\ref{upd:layer}) and the induction hypothesis, we have 
 \begin{equation}\label{pe:T}
 	\begin{aligned} 
 		\hat{\bm{e}}_{\pi(j),\pi(i)} = \bm{e}_{j,i}, \quad \hat{\bm{x}}_{\pi(i)}^{(s-1)} = \bm{x}_i^{(s-1)}, \\ \mathcal{N}(\pi(i) ) = \{\pi(j), j \in \mathcal{N}(i) \}.
 	\end{aligned}
 \end{equation}
 
 Plugging (\ref{pe:T}) into (\ref{upd:layer}), we have $\bm{x}_i^{(s)} = \hat{\bm{x}}_{\pi(i)}^{(s)}$. Thus, for an $S$-layer MPGNN, the output satisfies
 \begin{align}
 	\bm{x}_i^{(S)} = \hat{\bm{x}}_{\pi(i)}^{(S)}.
 \end{align}
 
 The output matrix of the original graph is $\bm{X} = [\bm{x}_1, \cdots, \bm{x}_{|V|}]^T$ and the output matrix of the permuted graph is $\hat{\bm{X}} = [\hat{\bm{x}}_1, \cdots, \hat{\bm{x}}_{|V|}]^T$. Thus $\bm{X}_{(i,:)} = \hat{\bm{X}}_{(\pi(i), :)}$ and we have
 \begin{align*}
 	\Phi((\pi \star \bm{Z}, \pi \star A)) = \hat{\bm{X}} = \pi \star \bm{X} = \pi \star \Phi(( \bm{Z}, A)).
 \end{align*}

 \section{Proof of Theorem \ref{thm:MPGNNMB}} \label{proof:MPGNNMB}
 In MB-DLAs, the maximal degree of nodes should be bounded by some constant, denoted by $\Delta$. The total number of iterations of MB-DLA and the number of layers of MPGNN denoted by $S$. The update of MB-DLA at the $t$-th iteration can be written as
 \begin{align} \label{upd:mb-dla}
 &\hat{\bm{m}}_{i,j}^{(t)} = h_2^{(t)}\left(h_1^{(t)}(\hat{\bm{x}}_i^{(t-1)}),\bm{e}_{i,j}\right), \notag\\
 &\hat{\bm{x}}^{(t)}_i = g_2^{(t)}\left(\hat{\bm{x}}^{(t-1)}_i,  g_1^{(t)}\left(\left\{\hat{\bm{m}}^{(t)}_{j,i}: j \in \mathcal{N}(i)\right\}\right)\right).
 \end{align}
 
 The update of an MPGNN at the $k$-layer can be written as
 \begin{align} \label{upd:mpgnn} 
 \bm{x}_i^{(k)} = \alpha^{(k)}\left(\bm{x}_i^{(k-1)}, \phi^{(k)} \left\{\left[\bm{x}_j^{(k-1)},\bm{e}_{j,i}\right]: j \in \mathcal{N}(i) \right\}   \right).
 \end{align}
 
 1) We first show that the inference stage of an MPGNN can be viewed as an MB-DLA, i.e., for all $0 \leq s \leq S$ and $\{\alpha^{(k)},\phi^{(k)}\}_{1 \leq k \leq s}$, there exists $\{h_2^{(t)},h_1^{(t)}, g_2^{(t)}, g_1^{(t)}\}_{1 \leq t \leq s}$ such that $\hat{\bm{x}}^{(s)}_i = \bm{x}_i^{(s)}$. We prove it by induction. The base case $\hat{\bm{x}}_i^{(0)} = \bm{x}^{(0)}_i$ holds because both $\hat{\bm{x}}_i^{(0)}$ and $\bm{x}^{(0)}_i$ are node features of the same node. We then assume $\hat{\bm{x}}^{(n)}_i = \bm{x}_i^{(n)}$ when $n = s - 1$. When $n = s$, we construct 
 \begin{align*}
 	\bm{y}_i = h_1^{(s)}(\hat{\bm{x}}_i^{(s-1)}) = \hat{\bm{x}}_i^{(s-1)}, h_2^{(s)}\left(\bm{y}_i,\bm{e}_{i,j}\right) = [\bm{y}_i,\bm{e}_{i,j}],
 \end{align*}
 and thus
 \begin{align*}
 	\hat{\bm{m}}_{i,j}^{(s)} = \left[\bm{x}_j^{(s-1)},\bm{e}_{j,i}\right].
 \end{align*}
 
 Let $g_2^{(s)} = \alpha^{(s)}$ and $g_1^{(s)} = \phi^{(s)}$, we have $\hat{\bm{x}}^{(s)}_i = \bm{x}_i^{(s)}$.

 2) We then show that an MB-DLA can be viewed as an MPGNN, i.e., for all $0 \leq s \leq S$ and $\{h_2^{(t)},h_1^{(t)}, g_2^{(t)}, g_1^{(t)}\}_{1 \leq t \leq s}$, there exists $\{\alpha^{(k)},\phi^{(k)}\}_{1 \leq k \leq s}$, such that $\hat{\bm{x}}^{(s)}_i = \bm{x}_i^{(s)}$. We prove it by induction. The base case $\hat{\bm{x}}_i^{(0)} = \bm{x}^{(0)}_i$ holds because both $\hat{\bm{x}}_i^{(0)}$ and $\bm{x}^{(0)}_i$ are node features of the same node. We then assume $\hat{\bm{x}}^{(n)}_i = \bm{x}_i^{(n)}$ when $n = s - 1$. We first define some notations. For a set of vectors $\mathcal{Z} = \{\bm{z}_1, \cdots \bm{z}_{\Delta} \}$, where $\bm{z}_i \in \mathbb{C}^n$, we define the order of variables in the set by the order of real part of its first coordinate. Let $\tau_i(\mathcal{Z}), 1 \leq i \leq \Delta$ denote the function that selects the $i$-th element in a multiset $\mathcal{Z}$. Denote $\bm{y}_{j,i} = \left[\bm{x}_j^{(s-1)},\bm{e}_{j,i}\right]$, $\mathcal{Y}_i = \{\bm{y}_{1,i}, \cdots \bm{y}_{\Delta,i} \}$ and define $\chi_{1}: \bm{x} \mapsto \bm{x}_{(1:d_1)}$, $\chi_{2}: \bm{x} \mapsto \bm{x}_{(d_1+1:d_2)}$. We construct 
 \begin{align*}
 \phi^{(s)} (\mathcal{Y}_i)  = g_2^{(s)}\left.(\left.\{ h_2\left(h_1(\chi_{1} (\tau_1(\mathcal{Y}_i) ) ), \chi_{2} (\tau_1(\mathcal{Y}_i) ) \right), \right.\right.\\
 \left.\left.\cdots, h_2\left(h_1( \chi_{1}(\tau_{\Delta}(\mathcal{Y}_i))  ), \chi_{2}(\tau_{\Delta}(\mathcal{Y}_i))\right)   \right.\}\right.),
 \end{align*}
 and $\alpha^{(s)} = g_1^{(s)}$, we then obtain $ \bm{x}^{(s)}_i = \hat{\bm{x}}_i^{(s)}$.
 
 \section{Proof of Proposition \ref{prop:wmmse}} \label{proof:wmmse}
 Here, we consider WMMSE \cite{Shi2011An} in the original paper's setting, which includes (\ref{eq:sys_mod}) as a special case. The WMMSE algorithm considers a $K$ cell interfering broadcast channel where base station (BS) $k$ serves $I_k$ users. Denote $\bm{H}_{i_k,j}$ as the channel from base station $j$ to user $i_k$, $\bm{V}_{i_k}$ as the beamformer that BS $k$ uses to transmit symbols to user $i_k$, $w_{i_k}$ as the weight of user $i_k$, and $\sigma^2_{i_k}$ as the variance of noise for user $i_k$. The problem formulation is
\begin{equation}
	\footnotesize
	\begin{aligned}
	&\underset{\bm{V}}{\text{maximize}}
	& & \sum_{k=1}^{K} \sum_{i=1}^{I_k} w_{i_k} \text{logdet}\left.(\bm{I} + \bm{H}_{i_k,k} \bm{V}_{i_k} \bm{V}_{i_k}^{H}\bm{H}_{i_k,k}^{H} \right. \\
	& & & \left. \left(\sum_{(l,j) \neq (i,k)}  \bm{H}_{i_k,j} \bm{V}_{l_j} \bm{V}_{l_j}^{H}\bm{H}_{i_k,j}^{H} + \sigma_{i_k}^2 \bm{I} \right)^{-1}\right.) \\
	& \text{subject to}
	& & \text{Tr}(\bm{V}_{i_k}\bm{V}_{i_k}^H) \leq P_{\text{max} }, \forall k, \notag
	\end{aligned}
\end{equation}

The WMMSE algorithm is shown in Algorithm \ref{alg:wmmse}. We first model this system as a graph. We treat the $i_k$-th user as the $i_k$-th node in the graph. The node features are $[w_{i_k},\sigma_{i_k},\bm{H}_{i_k,k}]$. The internal state of node $i_k$ at the $(p-1)$-th iteration is $[\bm{U}^{(p-1)}_{i_k},\bm{W}^{(p-1)}_{i_k},\bm{V}^{(p-1)}_{i_k}, w_{i_k},\sigma_{i_k}, \bm{H}_{i_k,k}  ]$. An edge is drawn from the $l_j$-th node to the $i_k$-th node for all $l$ if there is an interference link between the $j$-th BS and the $i_k$-th user. The edge feature of the edge $(l_j,i_k)$ is $\bm{e}_{l_j,i_k} = [\bm{H}_{i_k,j},\bm{H}_{l_j,k}]$.

We show that a WMMSE algorithm with $T$ iterations is an MB-DLA with \emph{at most} $2T$ iterations. In the corresponding MB-DLA, we update the variables $\bm{U}_{i_k}$ and $\bm{W}_{i_k}$ at the odd iterations while updating the variable $\bm{V}_{i_k}$ at the even iterations. Specifically, at the $p$-th iteration with $p$ being an odd number, the $l_j$-th node broadcasts its state $\bm{V}_{l_j}^{(p-1)}$ along its edges. The edge $(l_j,i_k)$ processes the message by forming $\bm{m}_{l_j,i_k} = \bm{H}_{i_k,j}\bm{V}^{(p-1)}_{l_j}\bm{V}_{l_j}^{(p-1) H}\bm{H}^{H}_{i_k,j}$ and the node $i_k$ receives the message set $\{\bm{m}_{l_j,i_k},\forall l,j \}$. The agent $i_k$ first sums over the messages $\bm{M}_{i_k} = \sum_{l,j} \bm{m}_{l_j,i_k}$. Then the $i_k$-th node updates its internal state as $\bm{U}_{i_k}^{(p)} = (\bm{M}_{i_k} + \sigma_{i_k}^2 \bm{I})^{-1} \bm{H}_{i_k,k} \bm{V}^{(p-1)}_{i_k}$ and $\bm{W}^{(p)}_{i_k} =  \left(\bm{I} - \bm{U}^{(p)H}_{i_k} \bm{H}_{i_k,k}\bm{V}^{(p-1)}_{i_k} \right)^{-1}$. Specifically, at the $p$-th layer, we construct
\begin{equation}
\small
\begin{aligned}
& \bm{x}^{(p-1)}_{l_j} = [\bm{U}_{l_j}^{(p-1)},\bm{W}_{l_j}^{(p-1)},\bm{V}_{l_j}^{(p-1)},w_{l_j},\sigma_{i_k},\bm{H}_{l_j,j}], \\ 
& \bm{e}_{l_j,i_k} = [\bm{H}_{i_k,j},\bm{H}_{l_j,k}], \\
& \bm{y}_{i_k} = h_1^{(p)}\left(\bm{x}_{i_k}^{(p-1)}\right) = \bm{x}_{i_k}^{(p-1)}, \\
& \bm{m}_{l_j,i_k} = h_2^{(p)}\left(\bm{y}_{i_k}, \bm{e}_{l_j,i_k} \right)= \bm{H}_{i_k,j}\bm{V}^{(p-1)}_{l_j}\bm{V}_{l_j}^{(p-1) H}\bm{H}^{H}_{i_k,j}, \\
&\bm{M}_{i_k} = g_1^{(p)}(\{\bm{m}_{l_j,i_k},\forall l,j \}) = \sum_{l,j} \bm{m}_{l_j,i_k}, \\
& \bm{U}_{i_k}^{(p)} =  (\bm{M}_{i_k} + \sigma_{i_k}^2 \bm{I})^{-1} \bm{H}_{i_k,k} \bm{V}^{(p-1)}_{i_k}, \\
& \bm{W}^{(p)}_{i_k} =  \left(\bm{I} - \bm{U}^{(p)H}_{i_k} \bm{H}_{i_k,k}\bm{V}^{(p-1)}_{i_k} \right)^{-1},  \\
&\bm{x}^{(p)}_{i_k} = g_2^{(p)}(\bm{x}_{i_k}^{(p-1)},\bm{M}_{i_k} \\
&= [\bm{U}_{i_k}^{(p)}, \bm{W}^{(p)}_{i_k}, \bm{V}^{(p-1)}_{i_k}, w_{i_k}, \sigma_{i_k}, \bm{H}_{i_k,k} ]. \notag
\end{aligned}
\end{equation}

At the $p$-th iteration where $p$ is even, the $l_j$-th node broadcasts its state $[\bm{V}_{l_j}^{(p-1)},\bm{W}_{l_j}^{(p-1)}]$ along its edges. The edge $(l_j,i_k)$ processes the message by forming $\bm{m}_{l_j,i_k} = \bm{H}^{H}_{l_j,k}\bm{U}^{(p-1)}_{l_j}\bm{W}^{(p-1)}_{l_j}\bm{U}_{l_j}^{(p-1)H}\bm{H}_{l_j,k}$. Node $i_k$ receives the message set $\{\bm{m}_{l_j,i_k},\forall l,j \}$. The agent $i_k$ first sums over the messages $\bm{M}_{i_k} = \sum_{l,j} \bm{m}_{l_j,i_k}$. Then the $i_k$-th node updates its internal state as \\$\bm{V}^{(p)}_{i_k} =  w_{i_k}\left(\bm{M}_{i_k} + \mu_{i_k}^* \bm{I}  \right)^{-1} \bm{H}_{i_k,k}^{H} \bm{U}^{(p-1)}_{i_k}\bm{W}_{i_k}^{(p-1)}$. Specifically, we construct

\begin{equation}
\small
	\begin{aligned}
	& \bm{x}^{(p-1)}_{l_j} = [\bm{U}_{l_j}^{(p-1)},\bm{W}_{l_j}^{(p-1)},\bm{V}_{l_j}^{(p-1)},w_{l_j},\sigma_{i_k},\bm{H}_{l_j,j}], \\
	& \bm{e}_{l_j,i_k} = [\bm{H}_{i_k,j},\bm{H}_{l_j,k}], \\
	& \bm{y}_{i_k} = h_1^{(p-1)}\left(\bm{x}_{i_k}^{(p-1)}\right) = \bm{x}_{i_k}^{(p-1)}, \\ 
	&  \bm{m}_{l_j,i_k} = h_2^{(p-1)}\left(\bm{y}_{i_k}, \bm{e}_{l_j,i_k} \right)  \\ &  = w_{l_j}\bm{H}^{H}_{l_j,k}\bm{U}^{(p-1)}_{l_j}\bm{W}^{(p-1)}_{l_j}\bm{U}_{l_j}^{(p-1)H}\bm{H}_{l_j,k}, \\
	&\bm{M}_{i_k} = g_1^{(p)}(\{\bm{m}_{l_j,i_k},\forall l,j \}) = \sum_{l,j} \bm{m}_{l_j,i_k}, \\
	& \hat{\bm{V} }_{i_k}(\mu_{i_k}) =  w_{i_k}\left(\bm{M}_{i_k} + \mu_{i_k} \bm{I}  \right)^{-1} \bm{H}_{i_k,k}^{H} \bm{U}^{(p-1)}_{i_k}\bm{W}_{i_k}^{(p-1)}, \\ 
	& \mu_{i_k}^* = \underset{\mu_{i_k} \geq 0, \text{Tr}(\hat{\bm{V} }_{i_k}^{H}  \hat{\bm{V} }_{i_k}) \leq P_{\text{max} }}{\text{argmin}} \quad \mu_{i_k}, \\
	& \bm{V}^{(p)}_{i_k} =  w_{i_k}\left(\bm{M}_{i_k} + \mu_{i_k}^* \bm{I}  \right)^{-1} \bm{H}_{i_k,k}^{H} \bm{U}^{(p-1)}_{i_k}\bm{W}_{i_k}^{(p-1)},\\
	&\bm{x}_{i}^{(p)} = g_1^{(p)}(\bm{x}_{i_k}^{(p-1)},\bm{M}_{i_k}) \\ 
	&= [\bm{U}_{i_k}^{(p-1)}, \bm{W}^{(p-1)}_{i_k}, \bm{V}^{(p)}_{i_k}, w_{i_k}, \sigma_{i_k}, \bm{H}_{i_k,k} ]. \notag
	\end{aligned}
\end{equation}

 This completes the proof for Proposition \ref{prop:wmmse}.
 
 \begin{algorithm}  
 	\scriptsize
 	\caption{WMMSE Algorithm \cite{Shi2011An}}  
 	\label{alg:wmmse}  
 	
 	\begin{algorithmic}[1]
 		\State Initialize $\bm{V}_{i_k}$.
 		\For{$t=0...T$}
 		\State $\bm{U}_{i_k} \leftarrow \left(\sum_{(l,j)} \bm{H}_{i_k,j}\bm{V}_{l_j}\bm{V}_{l_j}^{H}\bm{H}^{H}_{i_k,j}   + \sigma^2_{i_k}\bm{I} \right)^{-1}\bm{H}_{i_k,k} \bm{V}_{i_k}$
 		\State $\bm{W}_{i_k} \leftarrow \left(\bm{I} - \bm{U}_{i_k}^{H} \bm{H}_{i_k,k}\bm{V}_{i_k} \right)^{-1}$
 		\State $\bm{V}_{i_k} \leftarrow w_{i_k}\left(\sum_{l,j} w_{l_j}\bm{H}^{H}_{l_j,k}\bm{U}_{l_j}\bm{W}_{l_j}\bm{U}_{l_j}^{H}\bm{H}_{l_j,k} + \mu^*_{i_k} \bm{I}  \right)^{-1} \bm{H}_{i_k,k}^{H} \bm{U}_{i_k}\bm{W}_{i_k}$
 		\EndFor
 		\State Output $\bm{V}_{i_k}$.
 	\end{algorithmic}  
 \end{algorithm}
 
\section*{Acknowledgments}
The authors would like to thank anonymous reviewers and the editors for their constructive comments.

\bibliographystyle{ieeetr}
\bibliography{Reference}

\end{document}